\pdfoutput=1
\RequirePackage{etex}
\documentclass{article}

\usepackage{times}  
\usepackage{helvet}  
\usepackage{courier}  
\usepackage[hyphens]{url}  
\usepackage{graphicx} 
\usepackage{natbib}  
\usepackage{caption} 

\usepackage{algorithm}
\usepackage{algorithmic}

\usepackage{newfloat}
\usepackage{listings}
\DeclareCaptionStyle{ruled}{labelfont=normalfont,labelsep=colon,strut=off} 
\floatstyle{ruled}
\newfloat{listing}{tb}{lst}{}
\floatname{listing}{Listing}

\usepackage{latexsym}
\usepackage{amsthm}
\usepackage{amsmath}
\usepackage{comment}
\usepackage{color}
\usepackage{caption,subcaption}

\newtheorem{my-theorem}{Theorem}
\newtheorem{lemma}{Lemma}

\newtheorem{definition}{Definition}
\newtheorem{example}{Example}
\def\Mnatural{M$\sp{\natural}$}
\newcommand{\cvec}[1]{e_{#1}}

\makeatletter
\renewcommand{\ALG@name}{Mechanism}
\makeatother

\newcommand{\red}[1]{\textcolor{red}{#1}}
\newcommand{\yokoo}[1]{\textcolor{red}{\textbf{Yokoo says:} #1}}
\newcommand{\sun}[1]{\textcolor{blue}{\textbf{Sun says:} #1}}

\newcommand{\kimura}[1]{\textcolor{green}{\textbf{Kimura says:} #1}}
\newcommand{\my}[1]{#1}
\renewcommand{\cite}[1]{\citep{#1}}

    \title{Multi-Stage Generalized Deferred Acceptance Mechanism:
      Strategyproof Mechanism for Handling General Hereditary Constraints}
\author{Kei Kimura, Kwei-guu Liu, Zhaohong Sun, \\
Kentaro Yahiro, \& Makoto Yokoo\\
Kyushu University\\
\texttt{(kkimura@/liu@agent./zhaohong.sun@/}\\ \texttt{yahiro@agent./yokoo@)}\\ \texttt{inf.kyushu-u.ac.jp}}
\begin{document}

\maketitle

\begin{abstract}
The theory of two-sided matching has been extensively developed and applied to many real-life application domains. As the theory has been applied to increasingly diverse types of environments, researchers and practitioners have encountered various forms of distributional constraints.  Arguably, the most general class of distributional constraints would be hereditary constraints; if a matching is feasible, then any matching that assigns weakly fewer students at each college is also feasible. However, under general hereditary constraints, it is shown that no strategyproof mechanism exists that simultaneously satisfies fairness and weak nonwastefulness, which is an efficiency (students' welfare) requirement weaker than nonwastefulness. We propose a new strategyproof mechanism that works for hereditary constraints called the Multi-Stage Generalized Deferred Acceptance mechanism (MS-GDA).  It uses the Generalized Deferred Acceptance mechanism (GDA) as a subroutine, which works when distributional constraints belong to a well-behaved class called hereditary \Mnatural-convex set. We show that GDA satisfies several desirable properties, most of which are also preserved in MS-GDA. We experimentally show that MS-GDA strikes a good balance between fairness and efficiency (students’ welfare) compared to existing strategyproof mechanisms when distributional constraints are close to an \Mnatural-convex set.
\end{abstract}

\section{Introduction}
The theory of two-sided matching has been developed
and 
has been applied to many real-life application domains
(see 
\citet{Roth:CUP:1990} for a comprehensive survey in
this literature).
It has attracted considerable attention from AI researchers%
~\cite{aziz2022stable,Haris19matching,hosseini2015manipulablity,IsmailiHZSY19,kawase2017near,Yahiro18,suzuki2022strategyproof}.
As the theory has been applied to increasingly
diverse types of environments, researchers and practitioners have
encountered various forms of distributional constraints
(see 
\citet{ABY:aaai22:survery} for 
a comprehensive survey on various distributional constraints). 
Two streams of works exist on matching with distributional
constraints.
One stream scrutinizes constraints that arise from real-life applications,
such as regional maximum quotas \cite{kamakoji-basic},
individual/regional minimum quotas \cite{fragiadakis::2012,goto:17},
affirmative actions \cite{ehlers::2012,kurata:jair2017}), etc.
The other stream mathematically
studies an abstract and general class of constraints, such as those
that can be represented by a substitute choice function
\cite{Hatfield:AER:2005}, hereditary and \Mnatural-convex set constraints \cite{kty:2018}, and
hereditary constraints 
\cite{aziz:cutoff:2021,goto:17,kamakoji-concepts}.
This paper deals with hereditary constraints, 
which require 
that if a matching between students and colleges 
is feasible, then any matching that places
weakly fewer students at each college is also feasible.%
\footnote{Although our paper is described
in the context of a student-college matching problem, 
the obtained result is applicable to matching problems in general.}
When some distributional constraints are imposed,
there exists a trade-off between 
fairness (which requires that no student has justified envy)
and efficiency/students' welfare. 
In particular, 
\citet{cho:2022} show that under
hereditary constraints, no strategyproof mechanism can
simultaneously satisfy fairness and an efficiency requirement called weak nonwastefulness. 

Given this impossibility result, our goal is to develop 
a strategyproof mechanism that works for any hereditary constraints, and 
strikes a good balance between fairness and efficiency/students' welfare. 
If we can ignore either fairness or efficiency completely, there exist
two simple strategyproof mechanisms.
One is the 
Serial Dictatorship (SD) mechanism \cite{goto:17}, which 
is parameterized by an exogenous serial order over the students
called a master-list. 
Students are assigned sequentially according to the master-list. 
In our context with constraints, a student is assigned to 
her most preferred college among those such that doing so 
does not cause any constraint violation.
SD achieves nonwastefulness (which is stronger than
weak-nonwastefulness) and even (stronger) Pareto efficiency,
but the mechanism completely ignores colleges' preferences, thus violating fairness.
The other popular mechanism is the Artificial Cap Deferred Acceptance
mechanism (ACDA) \cite{goto:17}, which 
artificially lowers the maximum quota of
each college such that
any matching that satisfies the reduced maximum quotas also satisfies all distributional 
constraints. Then ACDA obtains a matching by applying 
the Deferred Acceptance (DA) mechanism \cite{Gale:AMM:1962} with respect to the reduced maximum quotas. This mechanism satisfies fairness, but the obtained matching can be extremely inefficient 
since the 
artificial maximum quotas must be determined independently from
students' preferences. 
In summary, SD can be extremely unfair, and ACDA can be extremely inefficient. 
This makes 
them undesirable in real application domains.

When distributional constraints are restricted to 
a well-behaved class called hereditary \Mnatural-convex set,
there exists a strategyproof mechanism called Generalized DA (GDA), 
which is fair and much more efficient than ACDA \cite{kty:2018}. However, when distributional constraints
do not form an \Mnatural-convex set, GDA is no longer
strategyproof. 
We show that we can construct 
a strategyproof mechanism based on GDA, which we 
call Multi-Stage Generalized 
Deferred Acceptance (MS-GDA).
We show that GDA satisfies several desirable properties, most of which are preserved in MS-GDA. 
There exists a strategyproof mechanism 
called Adaptive Deferred 
Acceptance mechanism (ADA) \cite{goto:17},
which can work for any hereditary constraints. 
We experimentally show that MS-GDA is fairer than 
ADA, while 
it does not sacrifice students' welfare too much when 
distributional constraints are close to \Mnatural-convex set.
\section{Model}
\label{sec:model}
A matching market under distributional constraints is given
by $(S, C, X, \succ_S, \succ_C, {f})$.
The meaning of each element is
as follows. 
\begin{itemize}
 \item $S=\{s_1, \ldots, s_n\}$ is a finite set of students.
 \item $C=\{c_1, \ldots, c_m\}$ is a finite set of colleges. 
Let $M$ denote $\{1, 2, \ldots, m\}$.
 \item $X \subseteq S\times C$ is a finite set of contracts.
Contract $x = (s, c) \in X$ 
represents the matching between student $s$  and college $c$.
\item For any $Y \subseteq X$, 
let $Y_s:=\{(s, c) \in Y \mid c \in
C\}$ and $Y_c:=\{(s, c) \in Y \mid s \in S\}$ denote the sets of contracts in $Y$ that involve $s$ and $c$, respectively.
\item 
$\succ_S = (\succ_{s_1}, \ldots, \succ_{s_n})$ is a profile of
the students' preferences. 
For each student $s$,  $\succ_{s}$ represents the
 preference of $s$
over $X_s \cup\{(s,\emptyset)\}$, where $(s, \emptyset)$ represents an outcome such that $s$ is unmatched. We assume $\succ_s$ is strict for each $s$.
We say contract $(s,c)$ is \emph{acceptable} for $s$ if
$(s, c) \succ_s (s, \emptyset)$ holds. 
We sometimes use notations like $c \succ_s c'$ instead of $(s,c) \succ_s (s,c')$.
\item $\succ_C = (\succ_{c_1}, \ldots, \succ_{c_m})$ is a profile of
the colleges' preferences. 
For each college $c$,  $\succ_{c}$ represents the preference of $c$
over $X_c \cup\{(\emptyset,c)\}$, where $(\emptyset,c)$ represents an outcome such that $c$ is unmatched. We assume $\succ_c$ is strict for each $c$.
We say contract $(s,c)$ is \emph{acceptable} for $c$ if
$(s, c) \succ_c (\emptyset, c)$ holds.
We sometimes write $s \succ_c s'$ instead of $(s,c) \succ_c (s',c)$.
\item ${f}: {\mathbf{Z}}_+^m \rightarrow \{-\infty, 0\}$ is a function that represents distributional constraints, 
where $m$ is the number 
of colleges and ${\mathbf{Z}}_+^m$ is 
the set of vectors of $m$ non-negative integers.
For $f$, we call a family of vectors 
$F=\{\nu \in {\mathbf{Z}}_+^m \mid f(\nu)=0 \}$ 
\emph{induced vectors} of $f$. 
\end{itemize}
We assume each contract $x$ in $X_c$ is acceptable for $c$.
This is without loss of generality because if 
some contract is unacceptable for a college, we can assume
it is not included in $X$.

We say $Y\subseteq X$ is 
a \emph{matching}, 
if for each $s \in S$, 
either (i) $Y_s=\{x\}$ and $x$ is acceptable for $s$,
or (ii) $Y_s = \emptyset$ holds.

For two $m$-element vectors $\nu, \nu' \in 
\mathbf{Z}_+^{m}$,
we say $\nu \leq \nu'$ if for all $i \in M$,
$\nu_i \leq \nu'_i$ holds.
We say $\nu < \nu'$ if $\nu \leq \nu'$ and
for some $i \in M$, 
$\nu_i < \nu'_i$ holds.
Also, let $|\nu|$ 
denote the $L_1$ norm of $\nu$: 
$|\nu|= \sum_{i\in M} \nu_i$. We sometimes call
$|\nu|$ the \emph{size} of $\nu$. 

\begin{definition}[feasibility with distributional constraints]
Let $\nu$ be a vector of $m$ non-negative
integers. We say $\nu$ 
is \emph{feasible} in $f$ if ${f}(\nu) = 0$.
For $Y \subseteq X$, let us define $\nu(Y)$ as
$(|Y_{c_1}|, |Y_{c_2}|, \ldots, |Y_{c_m}|)$.
We say $Y$ is \emph{feasible} (in $f$) if $\nu(Y)$ is feasible in ${f}$.
\end{definition}
In some model, each college $c_i$ is assumed to have its 
maximum quota / capacity limit 
$q_{c_i}$. In our model, we assume such capacity constraints are 
embedded in $f$.

%
Let us first introduce a very general class of constraints called
\emph{hereditary} constraints. 
Intuitively, heredity means that 
if $Y$ is feasible in $f$, 
then any subset $Y' \subset Y$ is also feasible in $f$. 
Let $\cvec{i}$ denote an $m$-element unit vector, where its
$i$-th element is $1$ and all other elements are $0$.
Let $\cvec{0}$ denote an $m$-element zero vector
$(0, \ldots, 0)$.
\begin{definition}[heredity]
We say a family of $m$-element vectors
$F\subseteq \textbf{Z}^m_+$ is 
\emph{hereditary} if $\cvec{0} \in F$ and 
for all 
$\nu, \nu' \in \mathbf{Z}_+^{m}$, 
if $\nu > \nu'$ and $\nu \in F$, 
then $\nu' \in F$ holds. 
We say ${f}$ is \emph{hereditary}
if its induced vectors 
are hereditary. 
\end{definition}

\citet{kty:2018} show that when 
$f$ is hereditary,  and its induced vectors satisfy 
one additional condition called \emph{\Mnatural-convexity}, 
there exists a general mechanism called Generalized Deferred Acceptance
mechanism (GDA), which satisfies several desirable properties.\footnote{%
To be more precise, 
\citet{kty:2018} show that to apply their framework, it is
necessary that the family of feasible matchings forms a matroid. 
When distributional constraints are defined on $\nu(Y)$ 
rather than on contracts $Y$, 
the fact that the family of feasible contracts forms a matroid
corresponds to the fact that (i) the family of feasible vectors forms an
\Mnatural-convex set, and (ii) it is
hereditary \cite{MS:dca:1999}.}

Let us formally define an \Mnatural-convex set. 
\begin{definition}[\Mnatural-convex set]
\label{def:mnatural}
We say a family of vectors $F \subseteq \mathbf{Z}^{m}_+$ 
forms an \emph{\Mnatural-convex set}, if for all $\nu, \nu' \in F$, 
for all $i$ such that $\nu_i > \nu'_i$, 
there exists $j \in \{0\}\cup \{k \in M \mid \nu_{k} <
\nu'_{k}\}$
such that $\nu - \cvec{i} + \cvec{j} \in F$ and 
$\nu' + \cvec{i} - \cvec{j} \in F$ hold.
We say ${f}$ satisfies \Mnatural-convexity 
if its induced vectors form an \Mnatural-convex set.
\end{definition}

An \Mnatural-convex set can be considered as a
discrete counterpart of a convex set in a continuous 
domain. Intuitively, Definition~\ref{def:mnatural} means that for two feasible vectors $\nu$ and $\nu'$, there exists another feasible vector, which is one step closer starting from $\nu$ toward $\nu'$, and vice versa.
An \Mnatural-convex set has been studied
extensively in discrete convex analysis, a branch of discrete mathematics. 
Recent advances in discrete convex analysis have found many applications
in economics (see the survey paper by \citet{murota:dca:2016}). 
Note that heredity and \Mnatural-convexity are independent properties. 

\citet{kty:2018} show that various real-life distributional constraints
can be represented as a hereditary \Mnatural-convex set. 
The list of applications includes matching markets
with regional maximum quotas \cite{kamakoji-basic}, individual/regional minimum quotas \cite{fragiadakis::2012,goto:17}, 
diversity requirements in school choice \cite{ehlers::2012,kurata:jair2017}), distance constraints \cite{kty:2018},
and so on. 
However, \Mnatural-convexity can be easily violated by introducing 
some additional constraints, as discussed in Section~\ref{sec:new-mechanism}.

With a slight abuse of notation, for two sets of contracts
$Y$ and $Y'$,
we denote $Y_s \succ_s Y'_s$ if either (i)
$Y_s = \{x\}$, $Y'_s = \{x'\}$, and $x \succ_s x'$
for some $x, x' \in X_s$,
or
(ii) $Y_s = \{x\}$ for some $x \in X_s$ that is  acceptable for $s$
and $Y'_s = \emptyset$.
Furthermore, we denote $Y_s \succeq_s Y'_s$ if  either $Y_s \succ_s Y'_s$
or $Y_s = Y'_s$.
Also, we use notations like  $x \succ_s Y_s$ or $Y_s
\succ_s x$, where $x$ is a contract and $Y$ is a matching.

Let us introduce several desirable properties of a matching and a
mechanism. Intuitively, nonwastefulness means that we cannot improve the matching of one student without hurting other students. Fairness means if student $s$ is rejected from college $c$, $c$ prefers all students accepted to it over $s$. 
We say a mechanism satisfies property A if the mechanism produces a matching that satisfies property A in every possible matching market.  
\begin{definition}[nonwastefulness]
\label{def:nonwastefulness}
In matching $Y$, student $s$ \emph{claims an empty seat} of college $c$
if 
$(s,c)$ is acceptable for $s$,
$(s,c) \succ_s Y_s$,
and $(Y \setminus Y_s) \cup \{(s, c)\}$ is feasible.
We say a matching $Y$ is \emph{nonwasteful} if no student claims an empty
seat. 
\end{definition}

\begin{definition}[fairness]
\label{def:fairness}
In matching $Y$, student $s$ \emph{has justified envy} toward
another student $s'$ if 
$(s,c)$ is acceptable for $s$,
$(s,c) \succ_s Y_s$, $(s', c) \in Y$, and
$(s, c) \succ_c (s', c)$ hold.
We say matching $Y$ is \emph{fair} if no student has justified envy.
\end{definition}

When additional distributional constraints (besides colleges' maximum quotas) 
are imposed, 
fairness and nonwastefulness become incompatible in general.
We can address this incompatibility by weakening the requirement of non-wastefulness. 
\citet{kamakoji-concepts} propose a weaker version of the nonwastefulness concept, which we refer to as \emph{weak nonwastefulness}.

\begin{definition}[weak nonwastefulness]
\label{def:weak-nonwastefulness}
In matching $Y$, student $s$ \emph{strongly claims an empty seat} of $c$
if 
$(s,c)$ is acceptable for $s$,
$(s, c) \succ_s Y_s$, 
and $Y \cup \{(s, c)\}$ is feasible. 
We say a matching $Y$ is \emph{weakly nonwasteful} if no student strongly claims an empty seat. 
\end{definition}
Clearly, if student $s$ can strongly claim an empty seat of $c$, 
she can also claim an empty seat of $c$, assuming $f$ is hereditary,
since if $Y \cup \{(s,c)\}$ is feasible, 
$(Y\setminus Y_s)\cup \{(s,c)\}$ is also feasible. 
Thus, 
nonwastefulness implies weak nonwastefulness but not 
vice versa.

We can also weaken fairness. Assume there exists a common strict 
ordering among students called master-list (ML). 
We denote the fact that $s$ is placed in a higher/earlier position than 
student $s'$ in master-list $L$ as $s \succ_{L} s'$. 
WLOG, we assume master-list $L$ is given as $(s_1, s_2, \ldots, s_n)$.
\begin{definition}[ML-fairness]
\label{def:ml-fairness}
Given matching $Y$ and master-list $L$,
we say $Y$ is ML-fair, if student $s$ has justified envy toward
another student $s'$, then $s' \succ_{L} s$ holds. 
\end{definition}
ML-fairness means 
the justified envy of student $s$ toward $s'$ is not considered legitimate if $s'$ is higher than $s$ in $L$. 

\begin{definition}[strategyproofness]
\label{def:strategyproofness}
We say a mechanism is \emph{strategyproof}
if no student ever has any incentive
to misreport her preference no matter what the other students report. 
More specifically, 
let $Y$ denote the matching obtained when $s$ declare her true preference $\succ_s$, 
and $Y'$ denote the matching obtained when $s$ declare something else, 
then $Y_s \succeq_s Y'_s$ holds. 
\end{definition}

We use the following example throughout the paper.
\begin{example}
\label{ex:common-example}
There are six students $s_1, \ldots, s_6$ and six colleges $c_1, \ldots, c_6$.
There are two regions:
$r_1=\{c_1, c_2, c_3\}$, $r_2=\{c_4, c_5, c_6\}$. 
For each region $r$, the total number of students assigned to the colleges in the region must be at most $q_r=3$.
Furthermore, colleges are categorized as rural or non-rural. 
We assume $c_3$ and $c_6$ are rural colleges.
We require that the total number of students assigned to non-rural colleges must be at most $q=4$.

Intuitively, we assume a policymaker wants to achieve the balance of the number of students allocated to each region. 
Also, the policymaker wants to restrict the number of students allocated to non-rural colleges in the hope that enough 
students will be assigned to rural colleges. 
Then, $f(\nu)=0$  holds when 
(i) $\nu_1 + \nu_2 + \nu_3\leq 3$, 
(ii) $\nu_4 + \nu_5 + \nu_6\leq 3$, and 
(iii) $\nu_1 + \nu_2 + \nu_4 + \nu_5 \leq 4$ hold. 

The preferences of all students are identical:
$c_1 \succ_s c_2 \succ_s c_4 \succ_s c_5 \succ_s c_3 \succ_s c_6$.
Also, the preferences of all colleges are identical:
$s_6 \succ_c s_5 \succ_c \ldots \succ_c s_1$. 
\end{example}
If only regional quotas exist, the above distributional constraints form a hereditary \Mnatural-convex set. 
By adding the quota for non-rural colleges (i.e., (iii)), 
\Mnatural-convexity will be violated. For example, 
both $\nu=(3,0,0,1,0,2)$ and $\nu'=(2,0,1,2,0,1)$ are feasible. 
Here, $\nu_1 > \nu'_1$ holds. \Mnatural-convexity requires that we can choose
$j \in \{0, 3, 4\}$ s.t. both $\nu - \cvec{1} + \cvec{j}$ and $\nu' + \cvec{1} - \cvec{j}$ are feasible. 
However, when $j=0$, $\nu' + \cvec{1} - \cvec{j}= (3, 0, 1, 2, 0, 1)$ is infeasible. 
When $j=3$, $\nu' + \cvec{1} - \cvec{j}=(3, 0, 0, 2, 0, 1)$ is infeasible. 
When $j=4$, $\nu' + \cvec{1} - \cvec{j}= (3, 0, 1, 1, 0, 1)$ is infeasible. 

To apply ACDA in Example~\ref{ex:common-example}, 
we must define the artificially reduced maximum quota of each college s.t. all distributional constraints are satisfied. 
This is possible by setting $q_c=1$. 
Then, the obtained matching is: 
$\{ (s_1, c_6), (s_2, c_3), (s_3, c_5), (s_4, c_4), 
\linebreak 
(s_5, c_2), (s_6, c_1)\}$. 
This matching is fair but too wasteful; only one student is assigned to the most popular college $c_1$. 

The obtained matching by SD is: 
$\{(s_1, c_1), (s_2, c_1), 
(s_3, c_1), (s_4, c_4), (s_5, c_6), (s_6, c_6)\}$.
This is nonwasteful but too unfair; students $s_4, s_5$ and $s_6$ have justified envy. 
\section{Existing Mechanisms}
\label{sec:existing-mechanisms}
\subsection{Generalized Deferred Acceptance (GDA)}
\citet{Hatfield:AER:2005} 
propose a generalized version 
of the Deferred Acceptance mechanism \cite{Gale:AMM:1962}, which we call Generalized Deferred Acceptance mechanism (GDA). 
To define GDA, we first introduce \emph{choice functions} of 
students and colleges.

\begin{definition}[students' choice function]
For each student $s$, her \emph{choice function} $Ch_s$ 
specifies her most preferred contract within each $Y \subseteq X$,
 i.e.,
$Ch_s(Y)= \{x\}$, where $x$ is the most preferred acceptable contract
in $Y_s$ 
if one exists, and $Ch_s(Y)=\emptyset$  if no such contract exists. 
Then, the choice function of all students is defined as
$Ch_S(Y) := \bigcup_{s \in S} Ch_s(Y_s)$.
\end{definition}

\begin{definition}[colleges' choice function]
\label{def:college:choice}
We assume each contract $(s,c) \in X$ is associated with 
its unique strictly positive weight $w((s,c))$. 
We assume these weights respect each college's preference
$\succ_c$, i.e., if $(s, c) \succ_c (s', c)$, 
then $w((s,c)) > w((s', c))$ holds. 
For $Y\subseteq X$, let $w(Y)$ denote $\sum_{x \in Y} w(x)$.
Then, the choice function of all colleges is defined as
$Ch_C(Y) := \arg\max_{Y'\subseteq Y} 
{f}(\nu(Y')) + w(Y')$.
\end{definition}
When distributional constraints are imposed, 
colleges' preferences are not sufficient to define 
colleges' choice function. 
For example, assume there are two colleges $c_1$ and $c_2$, 
and $f(\nu) = 0$ when $\nu_1+\nu_2\leq 1$, i.e., at most one 
students can be allocated to either $c_1$ or $c_2$. 
Further assume $Y=\{(s_1, c_1), (s_2, c_2)\}$. 
Then, $Ch_C$ must select one contract from $Y$. 
Since these two contracts are related to different colleges,
colleges' preferences $\succ_C$ do not tell us which one should be 
chosen. We assume the weight of each contract gives us the required 
information to make this decision. 

Note that the argument of $Ch_C$ can be any $Y\subseteq X$;
$Y$ does not have to be a matching, i.e., 
$|Y_s|$ can be strictly more than one. 
%
As long as ${f}$ induces a hereditary 
\Mnatural-convex set, a unique subset $Y'$ exists
that maximizes the above formula. 
Furthermore, such a subset can be efficiently 
computed in the following greedy way. 
Let $Y'$ denote the set of chosen contracts, which is initially $\emptyset$. 
Then, sort $Y$ in the decreasing order of their weights. 
Then, choose contract $x$ from $Y$ one by one and add it to 
$Y'$, as long as $Y'\cup\{x\}$ is feasible. 

Note that for simplicity, 
we use the simplest version of $Ch_C$ discussed in 
\cite{kty:2018}. The results obtained in this paper are
applicable to a more general class of choice functions, e.g., 
a choice function that tries to achieve a balanced allocation 
among colleges \cite{kty:2018}.

Using $Ch_S$ and $Ch_C$, GDA is defined as Mechanism~\ref{alg:gda}.
\begin{algorithm}[t]
\begin{algorithmic}[1]
\REQUIRE $X, Ch_S, Ch_C$
\ENSURE matching $Y$
\caption{Generalized Deferred Acceptance (GDA)}
\label{alg:gda}
\STATE $Re\leftarrow \emptyset$. 
\STATE Each student $s$ offers her most preferred contract $(s,c) $ which has not been rejected before
(i.e., $(s,c) \not\in Re$). If no remaining contract is acceptable  for $s$, $s$ does not make any offer.
Let $Y$ be the set of contracts offered (i.e., $Y=Ch_S(X\setminus Re)$). 
\STATE Colleges tentatively accept $Z=Ch_C(Y)$ and reject other contracts in $Y$ (i.e., $Y\setminus Z$).
\STATE If all the contracts in $Y$ are tentatively accepted at 3, then let $Y$ be the final matching 
and terminate the mechanism. Otherwise, $Re \leftarrow Re\cup (Y\setminus Z)$, and go to $2$.
\end{algorithmic}
\end{algorithm}
Note that we describe the mechanism using terms like "student $s$ offers" to make the description 
more intuitive. In reality, GDA is a direct-revelation mechanism, where the mechanism first collects 
the preference of each student and then the mechanism chooses a contract on behalf of each student. 

As its name shows, GDA is a generalized version 
of the well-known DA. 
In a basic model, 
the only distributional constraints are $(q_c)_{c \in C}$, i.e.,  
colleges' maximum quotas. 
Thus, ${f}(Y) = 0$ iff 
for each $c\in C$, 
$|Y_c|\leq q_c$ holds. 
Then, $Ch_C(Y)$ is defined as 
$\bigcup_{c\in C} Ch_c(Y_c)$, where
$Ch_c$ is the choice function of each college $c$, 
which chooses top $q_c$ contracts from $Y_c$ 
based on $\succ_c$. 

\citet{kty:2018} show that when $f$ induces a 
hereditary \Mnatural-convex set, GDA is strategyproof,
the obtained matching $Y$ satisfies a property called Hatfield-Milgrom
stability (HM-stability), and $Y$ is the student-optimal matching within all HM-stable matchings
(i.e., all students weakly prefer $Y$ over any other HM-stable matching). 

\begin{definition}[HM-stability]
Matching $Y$ is HM-stable if $Y=Ch_S(Y)=Ch_C(Y)$, and 
there exists no contract $x \in X\setminus Y$, such that 
$x \in Ch_S(Y\cup \{x\})$ and 
$x \in Ch_C(Y\cup \{x\})$ hold.
\end{definition}
Intuitively, HM-stability means there exists no contract in $X\setminus Y$ that is mutually preferred by students and colleges.
Note that HM-stability implies fairness. 
If student $s$ has justified envy in matching $Y$, there exists
$(s,c) \in X\setminus Y$, $(s', c) \in Y$, s.t. 
$(s,c) \succ_s Y_s$ and $w((s,c))> w((s',c))$ holds. 
Then, $(s,c) \in Ch_S(Y\cup \{(s,c)\})$ and 
$(s,c) \in Ch_C(Y\cup \{(s,c)\})$ hold, i.e., 
$Y$ is not HM-stable. 
Since fairness and nonwastefulness are incompatible under distributional constraints, HM-stability does not imply nonwastefulness.

We show that GDA also satisfies several desirable properties. 
As far as we know, we are the first to show that
GDA satisfies these properties. 
Let us introduce several concepts. 
First, we introduce a stronger requirement than fairness. 
\my{
\begin{definition}[Generalized justified envy]
\label{def:generalized-envy}
For matching $Y$, we say student $s$ has generalized justified envy 
toward $s'$ if there exist $(s,c) \in X\setminus Y$, 
$(s', c') \in Y$, $(s,c) \succ_s Y_s$,
$w((s,c))> w((s', c'))$, and
$(Y\setminus \{(s', c')\}) \cup \{(s,c)\}$ is feasible.
We say matching $Y$ is fair in terms of generalized justified envy if
no student has generalized justified envy. 
\end{definition}
Clearly, 
if student $s$ has justified envy toward $s'$ at some school $c$, $s$ also has generalized justified envy toward $s'$, but not vice versa. This is because the condition in 
Definition~\ref{def:fairness} 
is a special case of the condition 
in Definition~\ref{def:generalized-envy}
assuming $c=c'$. 
Note that in the above definition, it is possible that $s=s'$ holds, 
i.e., student $s$ can have generalized justified envy toward herself when $c \neq c'$ holds.}

\my{
Next, we introduce a concept called on the Pareto frontier, 
which means that a matching cannot be improved in terms of 
generalized justified envy and efficiency
simultaneously. 
For matching $Y$, let 
$Ev(Y)$ denote $\{(s,s') \mid s, s' \in S, 
s \text{ has } \linebreak 
\text{generalized justified envy toward } s' \text{ in } Y\}$.}

\begin{definition}[Pareto frontier]
\my{
We say matching $Y$ is weakly fairer (in terms of generalized justified envy)
than another matching $Y'$
if $Ev(Y)\subseteq Ev(Y')$ holds.
Furthermore, we say matching $Y$ weakly dominates another matching $Y'$, if 
$\forall s \in S, Y_s \succeq_s Y'_s$ and 
$\exists s \in S, Y_s \succ_s Y'_s$ hold.}

\my{
We say feasible matching $Y$ is on the Pareto frontier (in terms of generalized justified envy and efficiency) 
if there exists no feasible matching $Y'$ s.t. $Y'$ weakly dominates $Y$, and $Y'$ is weakly fairer than $Y$ in terms of generalized justified envy.}
\end{definition}
The following theorems hold.

\begin{my-theorem}
GDA is weakly nonwasteful when ${f}$ induces a hereditary 
\Mnatural-convex set. 
\end{my-theorem}
\begin{proof}
Assume, 
for the sake of contradiction, 
that 
GDA is not weakly nonwasteful. 
Then, for matching $Y$ obtained by GDA, 
there exists student $s$ who strongly claims an empty seat of $c$, i.e., $(s,c)$ is acceptable for $s$, 
$(s,c) \succ_s Y_s$, 
and $Y\cup \{(s,c)\}$ is feasible. 
Then, $(s,c) \in Ch_S(Y\cup \{(s,c)\}$ holds.
Also, since 
$Y\cup \{(s,c)\}$ is feasible, 
${f}(\nu(Y\cup \{(s,c)\}))=0$. 
Then, $Ch_C(Y\cup \{(s,c)\}) = \arg\max_{Y' \subseteq Y\cup \{(s,c)\}} f(\nu(Y')) + w(Y') 
= Y\cup \{(s,c)\}$ holds,
since we assume $w((s,c))$ is strictly positive. 
Thus, 
$(s,c) \in Ch_C(Y\cup \{(s,c)\})$ holds.
This contradicts the fact that GDA obtains an HM-stable matching.
\end{proof}

\begin{my-theorem}
\label{thm:fair-in-terms-of-GJE-gda}
GDA is fair in terms of generalized justified envy. 
\end{my-theorem}
\begin{proof}
\my{
Assume, for the sake of contradiction, that
GDA obtains $Y$ and student $s$ has justified envy toward $s'$ in terms of 
generalized justified envy, i.e.,
there exist $(s,c) \in X\setminus Y$, 
$(s', c') \in Y$, $(s,c) \succ_s Y_s$,
$w((s,c))> w((s', c'))$, and
$(Y\setminus \{(s', c')\})\cup \{(s,c)\}$ is feasible.
Then, $(s,c) \in Ch_S(Y \cup \{(s,c)\})$ holds since $(s,c) \succ_s Y_s$,
Also, $(s,c) \in Ch_C(Y\cup \{(s,c)\})$ holds since
$w((s,c))> w((s', c'))$, $f(Y)=0$, and 
$f((Y\setminus \{(s', c')\}) \cup \{(s,c)\})=0$ hold.
This violates the fact that GDA obtains an HM-stable matching. }
\end{proof}

\my{
To prove Theorem~\ref{lemma:GDA-PF},
we utilize several lemmas. 
\citet{kty:2018} shows that $Ch_C$ satisfies two properties: \emph{substitutability} and \emph{the law of aggregated demand}. 
Substitutability means for any two set of contracts $Y$ and $Y'$, where $Y' \subseteq Y$, 
$Y'\setminus Ch_C(Y') \subseteq Y \setminus Ch_C(Y)$ holds. Intuitively, this condition means if contract $x\in Y'$ is not chosen by $Ch_C(Y')$, 
then it is also not chosen in $Ch_C(Y)$. 
The law of aggregated demand means any two set of contracts $Y$ and $Y'$, where $Y' \subseteq Y$, 
$|Ch_C(Y')| \leq |Ch_C(Y)|$ holds, i.e., the number of chosen contracts weakly increases if more contracts are added. 
We show that $Ch_C$ satisfies the following property, which means that the order of application of $Ch_C$ is irrelevant to the obtained result.}

\begin{lemma}
\label{prop:ord-ir}\my{
For any $Y, Y_1, Y_2 \subseteq X$ s.t. $Y=Y_1 \cup Y_2$ and $Y_1 \cap Y_2 = \emptyset$,
$Ch_C(Y)= Ch_C(Y_1 \cup Ch_C(Y_2))$ holds.}
\end{lemma}
\begin{proof}\my{
Since $Ch_C$ is substitutable, we have 
\begin{align}
&Y \setminus Ch_C(Y) \supseteq Y_2 \setminus Ch_C(Y_2)\text{, and}\label{eq:sub1}\\
&Y \setminus Ch_C(Y) \supseteq (Y_1 \cup Ch_C(Y_2)) \setminus Ch_C(Y_1 \cup Ch_C(Y_2)).\label{eq:sub2}
\end{align}
Using these inclusion relations, we will show that $Ch_C(Y) \subseteq Ch_C(Y_1 \cup Ch_C(Y_2))$.
Let $x\in Y \setminus Ch_C(Y_1 \cup Ch_C(Y_2))$.
If $x \in Y_1 \cup Ch_C(Y_2)$, then $x \in (Y_1 \cup Ch_C(Y_2)) \setminus Ch_C(Y_1 \cup Ch_C(Y_2))$ and, from \eqref{eq:sub2}, $x\in Y \setminus Ch_C(Y)$ holds.
Otherwise (i.e., if $x \notin Y_1 \cup Ch_C(Y_2)$), we have $x \in Y_2 \setminus Ch_C(Y_2)$ since $Y = Y_1 \cup Ch_C(Y_2) \cup (Y_2 \setminus Ch_C(Y_2))$.
Then, from \eqref{eq:sub1}, $x\in Y \setminus Ch_C(Y)$ holds.
In both cases, we have $x\in Y \setminus Ch_C(Y)$.
Therefore, $Y \setminus Ch_C(Y_1 \cup Ch_C(Y_2)) \subseteq Y \setminus Ch_C(Y)$, 
implying that $Ch_C(Y) \subseteq Ch_C(Y_1 \cup Ch_C(Y_2))$.}

\my{
Now, from the law of aggregate demand, $|Ch_C(Y)| \geq |Ch_C(Y_1 \cup Ch_C(Y_2))|$.
Hence, $Ch_C(Y) = Ch_C(Y_1 \cup Ch_C(Y_2))$ holds.}
\end{proof}

\my{
Next, by using the fact that the application order of $Ch_C$ is irrelevant, 
we show that the output of GDA is the set of contracts that are chosen when 
all contracts offered during the execution of GDA are given at once.}
\begin{lemma}
\label{prop:GDA-opt}\my{
Let $X'$ be the set of contracts that are offered during the execution of GDA 
and $Y^*$ be the output of GDA.
Then $Y^*= Ch_C(X')$.}
\end{lemma}
\begin{proof}\my{
Let $Y_i$ be the set of contracts offered in the $i$-th iteration of Step 2 in GDA.
Clearly, $X' = Y_1 \cup \dots \cup Y_t$, where $t$ is the number of iterations.
We inductively show that $Ch_C(Y_i) = Ch_C(Y_1 \cup \dots \cup Y_i)$ for $i = 1,\dots,t$.
Once this is done, we have $Y^* = Ch_C(Y_t) = Ch_C(Y_1 \cup \dots \cup Y_t) = Ch_C(X')$.
For $i=1$, clearly, $Ch_C(Y_1) = Ch_C(Y_1)$.
For $i > 1$, assume that $Ch_C(Y_{i-1}) = Ch_C(Y_1 \cup \dots \cup Y_{i-1})$.
Let $Y'_i$ denote $Y_i \setminus (Y_1 \cup \dots \cup Y_{i-1})$, i.e., 
$Y'_i$ is the set of newly offered contracts at the $i$-th iteration.
Here, $Y_i = Ch_C(Y_{i-1}) \cup Y'_i$ holds since each student cannot offer a contract that has been rejected so far.
Then
\begin{align*}
Ch_C(Y_i) &= Ch_C(Ch_C(Y_{i-1}) \cup Y'_i)\\
&= Ch_C(Ch_C(Y_1 \cup \dots \cup Y_{i-1}) \cup Y'_i)\\
&= Ch_C(Y_1 \cup \dots \cup Y_{i-1} \cup Y'_i)\\
&= Ch_C(Y_1 \cup \dots \cup Y_{i-1} \cup Y_i),
\end{align*}
where the second equality holds by induction, the third equality holds by Lemma~\ref{prop:ord-ir}, and the fourth equality holds since $Y_i \setminus Y'_i =Ch_C(Y_{i-1}) \subseteq Y_{i-1}$.}
\end{proof}

\my{
Now, we are ready to prove Theorem~\ref{lemma:GDA-PF}.}
\begin{my-theorem}
\label{lemma:GDA-PF}
The matching obtained by GDA is on the Pareto frontier
when ${f}$ induces a hereditary \Mnatural-convex set. 
\end{my-theorem}
\begin{proof}
Let $Y^*$ be the output of GDA.
Assume, 
for the sake of contradiction, 
there exists another feasible matching $Y'$ that weakly dominates $Y^*$ and is weakly fairer than $Y^*$.
Since $Y^*$ is fair in terms of generalized justified envy by Theorem~\ref{thm:fair-in-terms-of-GJE-gda}, so is $Y'$.
Let $X'$ be the set of contracts that are offered during the execution of GDA.
Then, $Y' \subseteq X'$ holds, since $x \succeq_s Y^*_s$ for some $s\in S$ iff $x \in X'$ by the description of GDA.
Let us consider a situation where $X$ is restricted to $X'$. In this situation, GDA still obtains $Y^*$ 
since the contracts in $X \setminus X'$ do not affect the output of GDA. 
As noted in Section~\ref{sec:existing-mechanisms}, $Y^*$ is the student-optimal matching within all HM-stable matchings.
This, together with the fact that $Y'$ weakly dominates $Y^*$, implies that $Y'$ is not HM-stable, i.e., 
$Ch_S(Y') \neq Y'$, $Ch_C(Y') \neq Y'$, or 
there exists contract $x \in X' \setminus Y'$ such that 
$x \in Ch_S(Y'\cup \{x\})$ and 
$x \in Ch_C(Y'\cup \{x\})$ hold.
Since $Y'$ is a feasible matching, we have $Ch_S(Y') = Ch_C(Y') = Y'$.
Therefore, there exists contract $x =(s,c) \in X' \setminus Y'$ such that 
$x \in Ch_S(Y'\cup \{x\})$ and 
$x \in Ch_C(Y'\cup \{x\})$ hold.
Here, $x \in Ch_S(Y'\cup \{x\})$ implies that $x = (s,c) \succ_s Y'_s$.
Moreover, $x \in Ch_C(Y'\cup \{x\})$ implies that either (i) $Y'\cup \{x\}$ is feasible or (ii) $Y'\cup \{x\}$ is infeasible and there exists contract $y =(s',c') \in Y'$ such that $w(x) > w(y)$ and $(Y' \setminus \{y\}) \cup \{x\}$ is feasible.
However, case (i) does not occur.
To see this, recall that the set of feasible matchings forms a matroid.
Moreover, by Lemma~\ref{prop:GDA-opt}, 
$Y^*= Ch_C(X')$ holds, namely, $Y^* = \arg\max_{Y\subseteq X'}f(\nu(Y))+w(Y)$.
This implies that $Y^*$ is a maximal feasible matching.
Since the sizes of maximal feasible matchings are the same in any matroid, it follows that $|Y'| \le |Y^*|$, as $Y'$ is a feasible matching.
On the other hand, since $Y'$ weakly dominates $Y^*$, for each student $s \in S$ such that $Y^*_s \succ_s (s,\emptyset)$ we have $Y'_s \ge Y^*_s \succ_s (s,\emptyset)$, implying that $|Y'| \ge |Y^*|$.
Hence, $|Y'| = |Y^*|$ holds.
Thus, $Y'$ is also a maximal feasible matching, implying that $Y'\cup \{x\}$ is infeasible.
Hence, case (ii) occurs.
This implies that $s$ has a generalized justified envy towards $s'$, contradicting the fact that $Y'$ is fair in terms of generalized justified envy.
\end{proof}

To prove Theorem~\ref{theorem:weak-nonwasteful-igda} in the next
section, 
we use an alternative description of GDA 
defined as Mechanism~\ref{alg:gda-alt}.
\begin{algorithm}[t]
\begin{algorithmic}[1]
\REQUIRE $X, Ch_S, Ch_C$
\ENSURE matching $Y$
\caption{Alternative description of GDA}
\label{alg:gda-alt}
\STATE Choose $s \in S$.
\STATE Let $Z$ denote the matching and $Re$ denote the rejected contracts
obtained by Mechanism~\ref{alg:gda} for $X\setminus X_s$. $\hat{s} \leftarrow s$. 
\STATE Student $\hat{s}$ offers her most preferred contract $(\hat{s}, c_i)$ which has not been rejected before
(i.e., $(\hat{s},c_i) \not\in Re)$. $Y \leftarrow Z \cup \{(\hat{s},c_i)\}$.
\STATE If $Y$ is feasible, then let $Y$ be the final matching and 
terminate the mechanism. 
Otherwise, there exists exactly one rejected contract, i.e., 
$Ch_C(Y) = Y \setminus \{(s', c_j)\}$ holds. $Z\leftarrow Y \setminus \{(s', c_j)\}$, $Re\leftarrow Re \cup \{(s', c_j)\}$, 
and $\hat{s} \leftarrow s'$, go to 3. 
\end{algorithmic}
\end{algorithm}
\citet{kty:2018} show that 
when $Z=Ch_C(Z)$ holds, for any contract $x \in X\setminus Z$, 
$Ch_C(Z\cup\{x\})$ can be $Z$, $Z\cup\{x\}$, or 
$Z\cup\{x\}\setminus\{z\}$ for some $z\in Z$, i.e., 
either all contracts are accepted or exactly one contract
(which can be $x$ or another contract in $Z$) is rejected. 
Thus, at Line~4 in Mechanism~\ref{alg:gda-alt}, 
either all contracts in $Y$ are tentatively accepted or 
exactly one contract $(s', c_j)$ is rejected. 
Note that 
$s'$ can be $\hat{s}$ or another student. Also, 
$c_j$ can be $c_i$ or another college. When distributional constraints 
are imposed, it is possible that contract $(s', c_j)$ is 
rejected by adding $(\hat{s}, c_i)$, where $j\neq i$. 
For example, assume 
at most one student can be assigned to colleges $c_1$ and $c_2$. 
Then, when $Z=\{(s_1, c_1)\}$ and $Y=\{(s_1, c_1), (s_2, c_2)\}$, 
$Ch_C(Y)$ can be $\{(s_2, c_2)\}$, i.e., $(s_1, c_1)$ is rejected 
by adding $(s_2, c_2)$. 

To show that these two alternative descriptions are equivalent, 
we use the following lemma.
\begin{lemma}
\label{lemma:dummy}
For the original matching market, let us consider 
an extended market that includes $t$ dummy colleges 
$c_{d_1}, \ldots, c_{d_t}$. Any matching that assigns some student 
to these colleges is infeasible
(note that 
the distributional constraints of this extended market
form a hereditary \Mnatural-convex set as long as the original
distributional constraints form 
a hereditary \Mnatural-convex set). 
For student $s$, her preference can be modified such that 
a contract with a dummy college $(s, c_{d_i})$ can be inserted at any place. 

Then, the matching obtained by GDA (defined as Mechanism~\ref{alg:gda}) 
in this extended market
is exactly the same as the matching obtained by the GDA in the original 
market. 
\end{lemma}
\begin{proof}
In both cases, GDA obtains the unique student-optimal HM-stable matching. 
By the distributional constraints, any matching that contains
a contract between a student and a dummy college is infeasible in the 
extended market. Thus, the set of all feasible matchings, as well as 
the set of all HM-stable matchings, must be exactly the same in these 
two markets.
Thus, if $Y$ is the student-optimal HM-stable matching 
in the original market, it is also the student-optimal HM-stable matching 
in the extended market. 
\end{proof}

\begin{lemma}
Mechanisms~\ref{alg:gda} and \ref{alg:gda-alt} return exactly the same matching. 
\end{lemma}
\begin{proof}
Let case (i) and case (ii) 
denote the cases where a matching is obtained by 
Mechanism~\ref{alg:gda} and Mechanism~\ref{alg:gda-alt}, 
respectively. 
Let us consider another case, which we call case (iii). 
Here, we extend the market by 
adding $t$ dummy colleges $c_{d_1}, \ldots, c_{d_t}$
as Lemma~\ref{lemma:dummy}.
Also, the preference of student $s$ is modified s.t.
she prefers each dummy college over any original college. 
Let us assume a matching is obtained by Mechanism~\ref{alg:gda} in case (iii). 
By Lemma~\ref{lemma:dummy}, the matching obtained in case (iii) 
must be identical to case (i). 

Also, in case (iii), student $s$ applies to her original first-choice college only after she is rejected by all dummy colleges.
Let us assume $t$ is large enough 
(say, larger than $|X\setminus X_s|$ in the original market). 
Then, the assignment among students except for $s$
must be tentatively settled without $s$. 
Then, the remaining procedures in case (iii) after $s$ applies
to her original first-choice college, 
must be identical to those in case (ii) after Line~3 in Mechanism~\ref{alg:gda-alt}.
Thus, the matching obtained in case (ii) must be identical to 
that in case (iii).
Then, the obtained matchings in cases (i) and (ii) must be 
identical. 
\end{proof}

\subsection{Adaptive Deferred Acceptance (ADA)}
\citet{goto:17} propose a general algorithm
 called \emph{Adaptive
   Deferred Acceptance (ADA)} 
 that is strategyproof and nonwasteful when $f$ is hereditary. 
  ADA utilizes master-list $L$. 
It also utilizes the maximum quota of each college $q_{c_i}$, 
which is given as $\max_{\nu \mid f(\nu)=0} \nu_i$, 
i.e., the maximum number of students assigned to $c_i$
for any feasible matching. 
ADA works by iteratively invoking the standard DA.
 In ADA, college $c$ becomes forbidden if it cannot
 accommodate any additional student without 
 violating feasibility constraints,
  even though its 
 maximum quota is not reached yet.  
 Each stage of ADA comprises multiple rounds, 
 where students are added one by one according to the master-list, as long as no college becomes forbidden.
  When a college becomes forbidden, the mechanism
 finalizes the current matching and moves to the next stage
 with updated 
 maximum quotas.

 The formal description of ADA is given in Mechanism~\ref{algo:ada}.

 \begin{algorithm}[t]
     \begin{algorithmic}[0]
         \caption{Adaptive Deferred Acceptance (ADA)}
\label{algo:ada}         
\REQUIRE master-list $L=(s_1, s_2, \ldots, s_n)$,
$S, C, X, \succ_S, \succ_C, f$, and maximum quotas $q_C$
\ENSURE matching $Y$
\STATE \textbf{Initialization}: 
    $q^1_c \leftarrow q_c$ for each
$c \in C$, $Y \leftarrow \emptyset$. Go to 
\textbf{Stage}~1.
\STATE \textbf{Stage} $k$: Proceed to \textbf{Round} 1.
\STATE \textbf{Round} $t$: Select $t$ students from the top of $L$.
Let $Y'$ denote the matching obtained by the standard DA
for the selected students
under $(q^k_c)_{c \in C}$.
\STATE \rm{(i)}  If all students in $L$ are already selected,
then $Y\leftarrow Y \cup Y'$, output $Y$ and terminate the mechanism.

\STATE \rm{(ii)} If some college $c_i \in C$ is forbidden, 
i.e., $|Y'_{c_i}|<q^k_{c_i}$ and $f(\nu(Y\cup Y') + \cvec{i}) = -\infty$ hold, 
then $Y \leftarrow Y \cup Y'$.
Remove $t$ students from the top of $L$.
For each college $c$ that is forbidden, 
set $q^{k+1}_{c}$ to $0$.
For each $c \in C$, which is not forbidden,
set $q^{k+1}_c$ to $q^k_c - |Y'_c|$. 
Go to \textbf{Stage} $k+1$.
\STATE \rm{(iii)} 
Otherwise, go to \textbf{Round} $t+1$.
\end{algorithmic}
\end{algorithm}

To run ADA
in Example~\ref{ex:common-example}, we set the maximum quota of each college
$q_{c_i} = \max_{\nu \mid f(\nu)=0} \nu_i = 3$. 
At Round 3 of Stage 1, 
three students $s_1, s_2, s_3$ 
are tentatively assigned by DA; the obtained matching is
$\{(s_1, c_1), (s_2, c_1), (s_3, c_1)\}$.
Then, $c_2$ and $c_3$ are forbidden, since no more student can be assigned to them
(although no student is currently assigned to them).  
Thus, this matching is finalized. 
Next, $s_4$ is tentatively assigned to $c_4$.
Then, $c_4$ and $c_5$ are forbidden; we finalize this assignment. 
The current matching is:
$\{(s_1, c_1), (s_2, c_1), (s_3, c_1), (s_4, c_4)\}$.
All colleges except for $c_6$ are forbidden. 
Thus, the remaining students are assigned to it.
The final matching is: 
$\{(s_1, c_1), (s_2, c_1), (s_3, c_1), (s_4, c_4), (s_5, c_6), (s_6, c_6)\}$. 
This is identical to the matching obtained by SD.
Note that $q_c$ 
is not tight, i.e., 
whenever the number of students assigned to college $c$ reaches 
$q_c$, the number of students assigned to the region that includes $c$ also reaches its regional maximum quota. As a result, 
ADA fixes the current matching; students do not compete based on 
$\succ_C$.

\citet{goto:17} 
experimentally show that ADA is fairer
(i.e., fewer students have 
justified envy) than SD and more efficient than ACDA. 
However, we identify the following limitations of ADA:
(i) it cannot obtain a fair matching even when distributional constraints form 
a hereditary \Mnatural-convex set, 
while GDA is strategyproof, fair, and weakly nonwasteful, and
(ii) it can be very unfair when colleges' maximum quotas are not tight.
To be more precise, 
ADA is fairer than SD when students compete for 
the seats of a college based on the college's preference. 
When colleges' maximum quotas are not tight, 
ADA becomes quite similar to SD and too unfair since there is no competition among students. 
As a result, the mechanism becomes too unfair. 

Let us illustrate the major difference between ADA and MS-GDA.
As ADA repeatedly applies DA, 
MS-GDA uses GDA as a subroutine. 
ADA adds students one by one and obtains 
a matching ignoring distributional constraints except for colleges' maximum quotas. 
To make sure that the obtained matching is feasible, ADA checks whether
the next matching has a chance to be infeasible, i.e., for some college $c_i$, 
$f(\nu(Y\cup Y') + \cvec{i})=-\infty$ holds although $|Y'_{c_i}|< q^k_{c_i}$. If so, 
ADA fixes the current matching. 
If we use GDA instead of DA as a subroutine, we cannot apply a similar approach to ADA. 
Assume $Z$ is the original matching, and $Y$ is the matching obtained after introducing 
one more student. In DA, although $Z$ and $Y$ can be very different, 
$\nu(Z)$ and $\nu(Y)$ are quite similar, i.e., 
either $\nu(Z)=\nu(Y)$ holds, or for some $i\in M$, 
$\nu(Z) + \cvec{i}=\nu(Y)$ holds. Thus, it is possible to check whether the next matching has 
a chance to be infeasible, just by checking the neighbors of $\nu(Y)$. 
In GDA, $\nu(Z)$ and $\nu(Y)$ can be completely different. This is because, due to distributional constraints
in which  multiple colleges are involved, 
if the newly added student applies to college $c$ and gets accepted, 
another student can be rejected by another college $c'$, where $c \neq c'$. 
Then, this student applies to another college $c''$, and so on. 
Thus, we cannot apply an adaptive approach similar to ADA.
As a result, in MS-GDA, we first calculate positive integer $d$, which represents the number of students who can be 
assigned safely without violating distributional constraints, and assign $d$ students at once. 
\section{New Mechanism: Multi-Stage Generalized Deferred Acceptance (MS-GDA)}
\label{sec:new-mechanism}
This section describes our newly developed mechanism, which we call
Multi-Stage Generalized Deferred Acceptance (MS-GDA). 
In MS-GDA, we first calculate positive integer $d$, which represents the number of students who can be 
assigned safely without violating distributional constraints. More specifically, we choose $d$ s.t. 
if we restrict $|\nu|$ (i.e., the size of vector $\nu$) within $d$, 
the distributional constraints can be considered as 
an \Mnatural-convex set (s.t. we can apply GDA).
\citet{fragiadakis::2012} develop a mechanism 
called Multi-Stage Deferred Acceptance (MS-DA), which is based on 
a similar idea to MS-GDA. However, MS-DA uses DA as a subroutine and can handle only 
individual minimum quotas.

MS-GDA is defined as Mechanism~\ref{algo:agda}. 
It repeats several stages. In the $k$-th stage, it uses 
function $f_k: \mathbf{Z}^m_+ \rightarrow \{-\infty, 0\}$, which is derived from $f$. 
For $f_k$ and positive integer $d$, 
let $f^d_k: \mathbf{Z}^m_+ \rightarrow \{-\infty, 0\}$ denote
a function defined as follows:
\begin{equation*}
f^d_k(\nu) =
  \begin{cases}
    f_k(\nu) & \text{if $|\nu|\leq d$,} \\
    -\infty  & \text{otherwise.}
  \end{cases}
\end{equation*}


\begin{algorithm}[t]
     \begin{algorithmic}[0]
         \caption{Multi-Stage Generalized Deferred Acceptance (MS-GDA)}
\label{algo:agda}         
\REQUIRE master-list $L=(s_1, s_2, \ldots, s_n)$, 
$S, C, X, \succ_S, \succ_C$, 
and $f$.
\ENSURE matching $Y$
\STATE \textbf{Initialization}: 
$Y\leftarrow \emptyset$, $f_1 \leftarrow f$, 
go to \textbf{Stage}~1.
\STATE \textbf{Stage} $k$
\IF{no student remains in $L$}
\STATE terminate the mechanism and return $Y$. 
\ELSE \STATE  Choose positive integer $d$ s.t. $f^d_k$ induces an  \Mnatural-convex set. 
\STATE  Remove remaining top $d$ students from $L$ and assign 
them using GDA based on $f^d_k$. Let $Y^k$ denote the obtained matching. 
\STATE  
$Y\leftarrow Y \cup Y^k$. 
\STATE  Let $f_{k+1}$ denote a function s.t. 
$f_{k+1}(\nu) = f(\nu + \nu(Y))$ holds. 
\STATE  Go to \textbf{Stage} $k+1$. 
\ENDIF
\end{algorithmic}
\end{algorithm}

At each Stage~$k$, MS-GDA chooses positive integer $d$ s.t. 
$f^d_k$ induces an \Mnatural-convex set. 
Theorem \ref{theorem:existence-d} 
shows that we can always choose such $d$.
There can be multiple choices for $d$. In particular, 
if $f^d_k$ induces an \Mnatural-convex set, 
for any $d'$ s.t. $0< d' < d$, $f^{d'}_k$ also induces an \Mnatural-convex set. 
We will discuss how to choose appropriate $d$ 
later in this section. 

\begin{my-theorem}
\label{theorem:existence-d}
There exists $d\geq 1$ s.t. 
$f^d_k$ induces an \Mnatural-convex set. 
\end{my-theorem}
\begin{proof}
We show that we can choose $d=1$, i.e.,
$f^1_k$ always induces an \Mnatural-convex set.
Let $F_1$ denote the induced vectors of $f^1_d$, i.e., 
$F_1 = \{\nu \mid f_k(\nu) =0, |\nu| \leq 1\}$. 
By definition, $F_1 \subseteq \{\cvec{0}\}\cup \bigcup_{i\in M} \{\cvec{i}\}$
and $\cvec{0} \in F_1$. 
Then, for any two elements $v, v' \in F_1$, 
if one element $v$ is equal to $\cvec{i}$, 
and another element $v'$ is equal to $\cvec{0}$, 
$v_i > v'_i$ holds. 
Then, $v - \cvec{i} + \cvec{0}=\cvec{0} \in F_1$, 
and $v' + \cvec{i} - \cvec{0}=\cvec{i} \in F_1$ hold. 
If one element $v$ is equal to $\cvec{i}$, 
and another element $v'$ is equal to $\cvec{j}$, where $i\neq j$,  
$v_i > v'_i$ and $v_j < v'_j$ hold.
Then $v - \cvec{i} + \cvec{j}=\cvec{j} \in F_1$, 
and $v' + \cvec{i} - \cvec{j}=\cvec{i} \in F_1$ hold. 
Thus, $F_1$ forms an \Mnatural-convex set. 
\end{proof}

To run MS-GDA in Example~\ref{ex:common-example}, 
we need to define contract weights. 
Let us assume the weights are defined as:
$w((s_6, *)) > w((s_5, *)) > ... > w((s_1, *))$ holds, i.e., 
all contracts related to $s_6$ has a larger weight than 
any contracts related to $s_5$, and so on. 
In the first stage of MS-GDA, we can set $d=4$, since 
the inequality (iii) holds whenever $|\nu|\leq 4$. 
By assigning the first 4 students based on GDA, 
the obtained matching is:
$\{(s_1, c_4), (s_2, c_1), (s_3, c_1), (s_4, c_1)\}$.
Then, (iii) holds only when no more student is assigned to 
non-rural colleges, i.e., $f_2(\nu)=0$ only when 
$v_1=v_2=v_4 = v_5=0$ holds, which means that (iii) no longer violates 
\Mnatural-convexity. 
Thus, in the second stage, we can set $d=2$. 
Then, $s_5$ and $s_6$ are assigned to $c_6$. 
The final matching is: 
$\{(s_1, c_4), (s_2, c_1), (s_3, c_1), (s_4, c_1), (s_5, c_6), 
(s_6, c_6)\}$.
Here, $s_5$ and $s_6$ have justified envy. 

MS-GDA satisfies several required/desirable properties, i.e., 
feasibility, strategyproofness, weak nonwastefulness, ML-fairness, 
and on the Pareto frontier. 

\begin{my-theorem}
MS-GDA obtains a feasible matching. 
\end{my-theorem}
\begin{proof}
Let $Y$ denote the matching obtained by MS-GDA.
For each iteration of GDA, the choice function of each student
chooses at most one acceptable contract.
Also, each student is selected exactly once in the execution of MS-GDA. 
Thus, $Y$ is a matching, i.e., $Y$ contains at most one contract for 
each student, and each contract is acceptable for the relevant student. 
Next, we show that $Y$ is feasible. 
Assume MS-GDA terminates at Stage $t+1$.
Then, $Y= Y^1 \cup Y^2 \cup \ldots \cup Y^t$ holds. 
Since GDA obtains a feasible matching, 
$f^d_t(\nu(Y^t)) =0$ holds. Also, since at most 
$d$ students are assigned to colleges, 
$|\nu(Y^t)| \leq d$ holds.
Thus, $f_t(\nu(Y^t))=0$ holds. 
By definition, $f_t(\nu(Y^t)) = 
f(\nu(Y^t) +\nu(Y^1 \cup Y^2 \cup \ldots \cup Y^{t-1}))
= f(\nu(Y))$. Thus, $f(\nu(Y))=0$ holds. 
\end{proof}

\begin{my-theorem}
MS-GDA is strategyproof. 
\end{my-theorem}
\begin{proof}
Assume, for the sake of contradiction, that MS-GDA is not strategyproof, i.e., 
there exists student $s$, whose true preference is $\succ_s$. 
Let $Y$ denote the matching when $s$ declares $\succ_s$,
and $Y'$ denote the matching when $s$ declares 
$\widetilde{\succ}_s$. 
We assume $Y'_s \succ_s Y_s$ holds, i.e., $s$ is assigned to a better college by 
misreporting her preference. 
For student $s$, the stage that she is chosen is determined independently of her 
declared preference. Assume student $s$ is chosen at Stage $t$. 
Then, the matching obtained before Stage $t$ 
does not change even if $s$ modifies her preference. 
Thus, student $s$ is assigned to a better college by 
modifying her preference from $\succ_s$ to $\widetilde{\succ}_s$ in GDA
under $f^d_t$. 
However, since $f^d_t$ is hereditary and induces an 
\Mnatural-convex set, GDA must be strategyproof. 
This is a contradiction. 
\end{proof}

We next prove two lemmas, which are used in the proof of Theorem~\ref{theorem:weak-nonwasteful-igda}. The first 
(Lemma~\ref {lemma:sequence}) 
plays a pivotal role in demonstrating the correctness of the second (Lemma~\ref{lemma:first-stage}).
\begin{lemma}
\label{lemma:sequence}
Let $F$ be a hereditary \Mnatural-convex set,
and $v^1, v^2, \ldots, v^t$ be a sequence of $m$-element 
vectors, which satisfies 
the following properties.
\begin{enumerate}
    \item For each $k\in \{1, \ldots, t\}$, $v^k \in F$. 
    \item For each $k \in \{1, \ldots, t-1\}$,
    (i) $v^{k+1} = v^{k}+\cvec{k_1} - \cvec{k_2}$, 
    where $k_1, k_2 \in M$ and $k_1\neq k_2$, holds, 
    and 
    (ii) 
    $v^{k+1} + \cvec{k_2}=v^{k}+\cvec{k_1} \not\in F$ holds.
\end{enumerate}
Then, for each $\ell \in M$, if $v^{k+1} + \cvec{\ell} \in F$ holds,  
$v^k + \cvec{\ell} \in F$ holds.
\end{lemma}
\begin{proof}
Assume $v' = v^{k+1} + \cvec{\ell} \in F$ holds. 
Then $\ell\neq k_2$, since 
    $v^{k+1} + \cvec{k_2} \not\in F$ holds.
For $v'$ and $v^k=v^{k+1} -\cvec{k_1}+\cvec{k_2}$,
$v'_{\ell}> v^{k}_{\ell}$ holds. 
$v'_j < v^k_j$ holds only when $j=k_2$.
Since $F$ is an \Mnatural-convex set, 
there exists $j \in \{0, k_2\}$, where
$v' - \cvec{\ell} + \cvec{j} \in F$ 
and $v^k + \cvec{\ell} - \cvec{j} \in F$ hold. 
If $j=k_2$,
$v' - \cvec{\ell} + \cvec{k_2} = v^{k+1} + \cvec{k_2} \not\in F$
holds. Thus, $j$ must be $0$. 
Then $v^k+ \cvec{\ell} \in F$ holds.
\end{proof}

To show Theorem~\ref{theorem:weak-nonwasteful-igda}, i.e., 
MS-GDA is weakly nonwasteful, we utilize 
the following lemma. 
This lemma shows that when student $s$ is assigned
in the first stage of MS-GDA, $s$ never strongly claims
an empty seat. 
Note that the fact GDA is weakly nonwasteful 
does not directly imply this lemma since GDA applied in the first stage of MS-GDA
uses $f^d_1$, which is different from $f$. 
More specifically, for $\nu$ s.t. $|\nu|=d+1$, by definition, 
$f^d_1(\nu) = -\infty$ holds, while $f(\nu)$ can be $0$. 
\begin{lemma}
\label{lemma:first-stage}
Assume in the Stage~1 of MS-GDA, $S'\subseteq S$
students are assigned (where $|S'|=d$), and 
the obtained matching at Stage 1 is $Y$. 
Then, no student in $S'$ strongly claims an empty seat in $Y$. 
\end{lemma}
\begin{proof}
Assume, for the sake of contradiction, $s \in S'$
strongly claims an empty seat of $c_{\ell}$ in $Y$, i.e., 
$(s, c_{\ell})$ is acceptable for $s$, 
$(s, c_{\ell}) \succ_s Y_s$, and $f(\nu(Y) + \cvec{\ell})=0$ hold. 

There are two possible cases: (i) $|Y|<d$, and (ii) $|Y|=d$. 
In case (i), $f^1_d(\nu(Y)+\cvec{\ell}) = f(\nu(Y)+\cvec{\ell}) = 0$ holds
since $f^1_d(\nu) = f_d(\nu) = f(\nu)$ holds for any $\nu$ s.t. $|\nu|\leq d$. 
Then, for $Z=Y\cup\{(s,c_{\ell})\}$, 
$(s, c_{\ell}) \in Ch_S(Z)$ holds since $(s, c_{\ell}) \succ_s Y_s$ holds. 
Also, $Ch_C(Z)=Z$ holds for $Ch_C$ based on $f^1_d$, 
since $Z$ is feasible in $f^1_d$. Thus, $(s, c_{\ell}) \in Ch_C(Z)$ holds. 
However, this contradicts the fact that $Y$, which is obtained by GDA, must be HM-stable. 

For case (ii), let us examine the execution of Mechanism~\ref{alg:gda-alt}, where
students $S'\setminus\{s\}$ are first assigned, then $s$ offers her most preferred contract, and so on. 
More specifically, let us examine how $\nu(Z)$, where $Z$ is the tentatively accepted contracts, 
transits until Mechanism~\ref{alg:gda-alt} terminates. 
Initially, $Z$ is the matching obtained by assigning $S'\setminus\{s\}$. 
Since we assume $|Y|=d$, $|Z|=d-1$ holds.
In Step 3 of Mechanism~\ref{alg:gda-alt},
for $Z$, 
one contract 
$(\hat{s}, c_i)$ is added, and $(s', c_j)$ is rejected/removed (or no contract is rejected, and the mechanism terminates). 
When $i=j$, $\nu(Z)$ does not change. 
When $i\neq j$, the $i$-th element is increased by one, and the $j$-th element is decreased by one
in $\nu(Z)$. Let $\nu$ denote the original vector, and $\nu'$ denote the new vector. 
$\nu + \cvec{i} - \cvec{j} = \nu'$ holds. Also, since $(s', c_j)$ is rejected, 
$\nu+\cvec{i}= \nu'+\cvec{j}$ must be infeasible (both in $f^1_d$ and $f$). 
After repeating such changes, eventually, $Z\cup\{(\hat{s}, c_i)\}$ becomes feasible and 
Mechanism~\ref{alg:gda-alt} terminates. 
The obtained matching is $Y = Z\cup \{(\hat{s}, c_i)\}$.
Since we assume $Y \cup \{(s, c_{\ell})\}$ is feasible,
$Z \cup \{(s, c_{\ell})\}$ is also feasible.
Then, by repeatedly applying Lemma~\ref{lemma:sequence} for the sequence of vectors that represents
the trace of $\nu(Z)$, we obtain that, for each vector $\nu^k$ that appears in the sequence, 
$\nu^k + \cvec{\ell}$ is feasible. 
However, in the execution of Mechanism~\ref{alg:gda-alt}, 
there must exists a situation where $(s, c_{\ell})$ is rejected, i.e., 
for some $\nu^k$, $\nu^k + \cvec{\ell}$ must be infeasible. 
This is a contradiction. 
\end{proof}

\begin{my-theorem}
\label{theorem:weak-nonwasteful-igda}
MS-GDA obtains a weakly nonwasteful matching. 
\end{my-theorem}
\begin{proof}
Assume, for the sake of contradiction, student $s$, who is assigned 
at Stage~$k$, strongly claims an empty seat in MS-GDA. 
By Lemma~\ref{lemma:first-stage}, $k$ cannot be 1. 
Let $S'$ denote students assigned at stages $1, \ldots, k-1$, and
$Y'$ denote the matching obtained after Stage $k-1$.  
Then, let us consider another (reduced) market where 
$S$ is changed to $S\setminus S'$, and 
$f$ is changed to $f^k$, i.e., all students $S'$ are removed, 
while the distributional constraints are modified in the way that 
$Y'$ is already fixed. Then, it is clear that $s$ strongly claims an empty seat 
in MS-GDA in this new market, and $s$ is assigned in the first stage. 
However, this contradicts Lemma~\ref{lemma:first-stage}, which shows 
that no student assigned in the first stage strongly claims an empty seat. 
\end{proof}

\begin{my-theorem}
MS-GDA is ML-fair. 
\end{my-theorem}
\begin{proof}
Assume, for the sake of contradiction, 
student $s$, who is assigned in Stage~$k$, has justified envy toward another 
student $s'$, while $s \succ_{L} s'$ holds. 
Then, $s'$ must be assigned in State~$k'$, where $k'\geq k$ holds. 
Since GDA is fair, student $s$ never has justified envy toward 
another student who is assigned in the same stage. Thus, $k'>k$ holds. 
Assume $s'$ is assigned to college $c_{\ell}$. 
Let $Y'$ denote the matching obtained at the end of Stage~$k$. 
Student $s$ prefers $(s, c_{\ell})$ over $Y'_s$. 
We can use a similar argument as the proof of Theorem~\ref{theorem:weak-nonwasteful-igda}
to create a situation where $s$ has justified envy when $s$ is assigned 
in the first stage, and the obtained matching in the first stage is $Y'$. 
Then, due to Lemma~\ref{lemma:first-stage}, 
$Y'\cup\{(s, c_{\ell})\}$ is infeasible, i.e., 
$f(\nu(Y') + \cvec{\ell})=-\infty$. 
Thus, no more student can be assigned to $c_{\ell}$ in later stages. 
This contradicts our assumption that $s'$ is assigned to $c_{\ell}$.
\end{proof}

\begin{my-theorem}
The matching obtained by MS-GDA is on the Pareto frontier. 
\end{my-theorem}
\begin{proof}\my{
Let $Y$ be the matching obtained by MS-GDA. 
Assume, for the sake of contradiction, that matching $Y'$
weakly dominates $Y$, and weakly fairer than $Y$.
Let $k$ be the first stage in the MS-GDA s.t. there exists
a student who is matched at Stage $k$ of MS-GDA and strictly
prefers $Y'$ over $Y$. Then, because of the choice of $k$,
for all students assigned at any Stage $k'<k$, 
their assignments in $Y$ must be identical to $Y'$. 
Let $S_k$ denote the set of students allocated in Stage $k$. 
Also, let $Y_{S_k}$ denote $\bigcup_{s \in S_k} Y_s$
and $Y'_{S_K}$ denote $\bigcup_{s \in S_k} Y'_s$. 
Assume the reduced market where only students $S_k$ are present and,
distributional constraints are given as $f_k$.
Then, MS-GDA in this reduced market obtains $Y_{S_k}$
assuming $d$ is chosen as $|Y_{S_k}|$, 
$Y'_{S_k}$ weakly dominates $Y_{S_k}$, and $Y'_{S_k}$ is weakly fairer than $Y_{S_k}$. 
Also, $Ev(Y_{S_k})=\emptyset$ holds. 
By Theorem~\ref{lemma:GDA-PF}, 
the fact that $Y'_{S_k}$ weakly dominates $Y_{S_k}$ means
there exists a pair $(s, s')$, where $s, s' \in S_k$ and
$(s, s') \in Ev(Y'_{S_k})$. Then, $Y'_{S_k}$ cannot be weakly fairer than 
$Y_{S_k}$. This is a contradiction. }
\end{proof}

\begin{my-theorem}
\my{
Assuming the computation of $f$ and $d$ can be done in polynomial time, MS-GDA runs in polynomial time.}
\end{my-theorem}
\begin{proof}\my{
Since $d\geq 1$ holds, 
the number of required stages of MS-GDA is at most $n$. 
\citet{kty:2018} shows that GDA runs in 
$O(T(f)\cdot n^2)$, where $T(f)$ is the required time to calculate $f$. 
Thus, assuming the computation of $f$ and $d$ can be done in polynomial time, MS-GDA runs in 
polynomial time.}
\end{proof}

For choosing $d$ s.t. $f^d_k$ induces an \Mnatural-convex set, 
we can always choose $d=1$. However,
choosing larger $d$ is desirable to obtain a fairer matching, 
since more students can compete based on 
colleges' preferences. 
Actually, MS-GDA becomes equivalent to SD if $d=1$ in all stages.
For smaller $d$, we can simply enumerate induced 
vectors and check whether they form an \Mnatural-convex set. 
However, this becomes intractable for larger $d$. 
Let us illustrate two cases where we can obtain appropriate $d$. 
In Section~\ref{sec:evaluation}, we use two markets that 
correspond to these cases.

The first case is that 
$f$ is divided into two parts, i.e., 
$f(\nu) := g(\nu) + h(\nu)$, 
where $g$ represents base-line distributional constraints, 
which induces an \Mnatural-convex set, and
$h$ represents some additional constraints, which violate
\Mnatural-convexity. Further assume that $h$ has a simple structure, e.g., 
it is represented by a set of linear inequalities. 
In Example~\ref{ex:common-example}, we can assume (i) and (ii) correspond to $g$ and (iii) corresponds to $h$, respectively. 
Then, at Stage $k$, 
$f_k$ is also divided into two parts, i.e., 
$f_k(\nu) := g_k(\nu) + h_k(\nu)$, and $h_k$ has a simple structure. 
Thus, it is easy to find maximum $d$ where $h_k(\nu)=0$ holds for any $\nu$ s.t. $|\nu|\leq d$, i.e., 
$f^d_k$ induces an \Mnatural-convex set. 
Also,
$h_k$ can degenerate into several very simple constraints
such that they never violate \Mnatural-convexity, as shown in the application of 
MS-GDA to Example~\ref{ex:common-example}.

\my{%
The second case is that $f$ is represented as a disjunction of two conditions.\footnote{ 
We can easily extend to the cases where 
$f$ is a disjunction of more than two conditions.}
Let us assume $f$ is defined as, 
$f(\nu)=0$ iff $\widehat{g}(\nu)+h(\nu)=0$ or 
$g(\nu)+\widehat{h}(\nu)=0$. Equivalently, we can write
$f(\nu) = \max(\widehat{g}(\nu)+h(\nu), g(\nu) + \widehat{h}(\nu))$.
Also assume each of $\widehat{g} + \widehat{h}$, $\widehat{g}+h$, 
and $g+\widehat{h}$ induces a
 hereditary \Mnatural-convex set, and for all $\nu$, 
 $\widehat{g}(\nu)\leq g(\nu)$ and  $\widehat{h}(\nu)\leq h(\nu)$ hold. 
Let us show an example, which we call flexible quotas. 
Assume there are two colleges $c_1$ and $c_2$.
The baseline maximum quota of each college is 10. 
Furthermore, we can assign 3 more students to 
exactly one college, i.e., we have some flexibility in 
maximum quotas. 
Then, we can represent flexible quotas by assuming
$\widehat{g}(\nu)=0$ iff $\nu_1 \leq 10$, 
$g(\nu)=0$ iff $\nu_1 \leq 13$, 
$\widehat{h}(\nu)=0$ iff $\nu_2 \leq 10$, and 
$h(\nu)=0$ iff $\nu_2 \leq 13$.}

\my{
Note that $f$ does not induce an \Mnatural-convex set, since 
\Mnatural-convexity is not closed under union (which corresponds to a disjunctive condition). 
Let $\widehat{f}(\nu)$ denote $\widehat{g}(\nu) + \widehat{h}(\nu)$, i.e., 
$\hat{f}(\nu)=0$ iff $\widehat{g}(\nu)=0$ and $\widehat{h}(\nu)=0$ hold. 
Here, we consider a conjunction of stricter conditions and 
$\widehat{f}(\nu)$ induces a hereditary \Mnatural-convex set.}

\my{%
During the execution of MS-GDA, we first make sure that 
the obtained matching $Y$ satisfies $\widehat{f}(\nu(Y))=0$. More specifically, in the first stage, we 
choose $d^*$ such that 
for all $\nu \in {\mathbf{Z}}_+^m$ where $|\nu|\leq d^*$ 
and $f(\nu)=0$ hold, 
$\widehat{f}(\nu)=0$ also holds.
In the above flexible quota example, we can choose 
$d^*=10$. 
If $d^*=0$, which means that 
either $\widehat{g}$ or $\widehat{h}$ can be violated by adding one more student, 
we set $d$ to $1$ and assign one student. 
Then, we repeat the same procedure; 
in stage $k$, we choose $d^*$
for $f_k(\nu) = \max(\widehat{g}_k(\nu) + h_k(\nu), 
g_k(\nu)+ \widehat{h}_k(\nu))$ 
and $\widehat{f}_k(\nu) = \widehat{g}_k(\nu) + \widehat{h}_k(\nu)$,
until either $\widehat{g}_k(e_0)=-\infty$ or 
$\widehat{h}_k(e_0)=-\infty$ holds. 
If $\widehat{g}_k(e_0)=-\infty$ (or $\widehat{h}_k(e_0)=-\infty$) holds, 
$\widehat{g}$ (or $\widehat{h}$) becomes violated.
Then, we 
set $f_k= g_k + \widehat{h}_k$ 
(or $f_k= \widehat{g}_k + h_k$) and assign all remaining students at this stage. 
This means that 
we commit to $g + \widehat{h}$ (or $\widehat{g}+h$).
In the flexible quota example, assume 6 students are 
allocated to $c_1$, and 4 students are allocated to $c_2$. 
Then, in the next stage, we can allocate $4$ students, 
and so on. When 10 students are allocated to one college,
$d^*$ becomes 0 and we can allocate one student. 
Assume in stage $k$, 11 students are allocated to $c_1$ 
in total. Then, we need to set the maximum quota of 
$c_1$ to 13 (and the maximum quota of $c_2$ to 10). 
We assign all remaining students at once
assuming $f_k=g_k + \widehat{h}_k$, i.e., the total number of students 
allocated to $c_1$ is at most 13, while the total number of students allocated to $c_2$ is at most 10.}

\section{Experimental Evaluation}
\label{sec:evaluation}

Regarding theoretical properties, 
ADA is better than MS-GDA, and SD is better than ADA; all of them satisfy ML-fairness, 
MS-GDA is weakly nonwasteful, ADA is nonwasteful, and SD is Pareto efficient.
In this section,
we experimentally show
that MS-GDA strikes a good balance between fairness and efficiency compared to ADA/SD.
We use the following two markets in our experiments.
The first market (Market 1) corresponds to the first case
described in the end of Section~\ref{sec:new-mechanism}:
$n=1000$ students, $m=100$ colleges, $|R|=20$ regions with regional quota $q_r=50$. In each region, 
one rural college and 4 non-rural colleges exist. 
We set the limit/quota $q$ for the number of students allocated to non-rural 
colleges. 
\my{%
The second market (Market 2) corresponds to the second
case, where we introduce flexibility in regional 
quotas: $n=1000$ students, $m=200$ colleges, $|R|=20$ regions with regional quota $q_r=60$. Each region has 10 colleges. The maximum quota of 
each college $q_c$ is set to $10$. 
The regions $r_1$ to $r_{10}$ (as well as $r_{11}$ to $r_{20}$) form East (and West) regions. 
The baseline regional quotas of these two large regions
are $450$. In addition, at most one large region (i.e., East or West) can accept $100$ more students.}

\def\trimbeforecaption{-0cm}
\begin{figure}[t]
	\captionsetup[subfigure]{aboveskip=5pt,belowskip=2pt}
	\centering
	\captionsetup{justification=centering}
    \subfloat[Effect of $\phi$ ($q = 800$)]
	{\includegraphics[scale = 0.23]{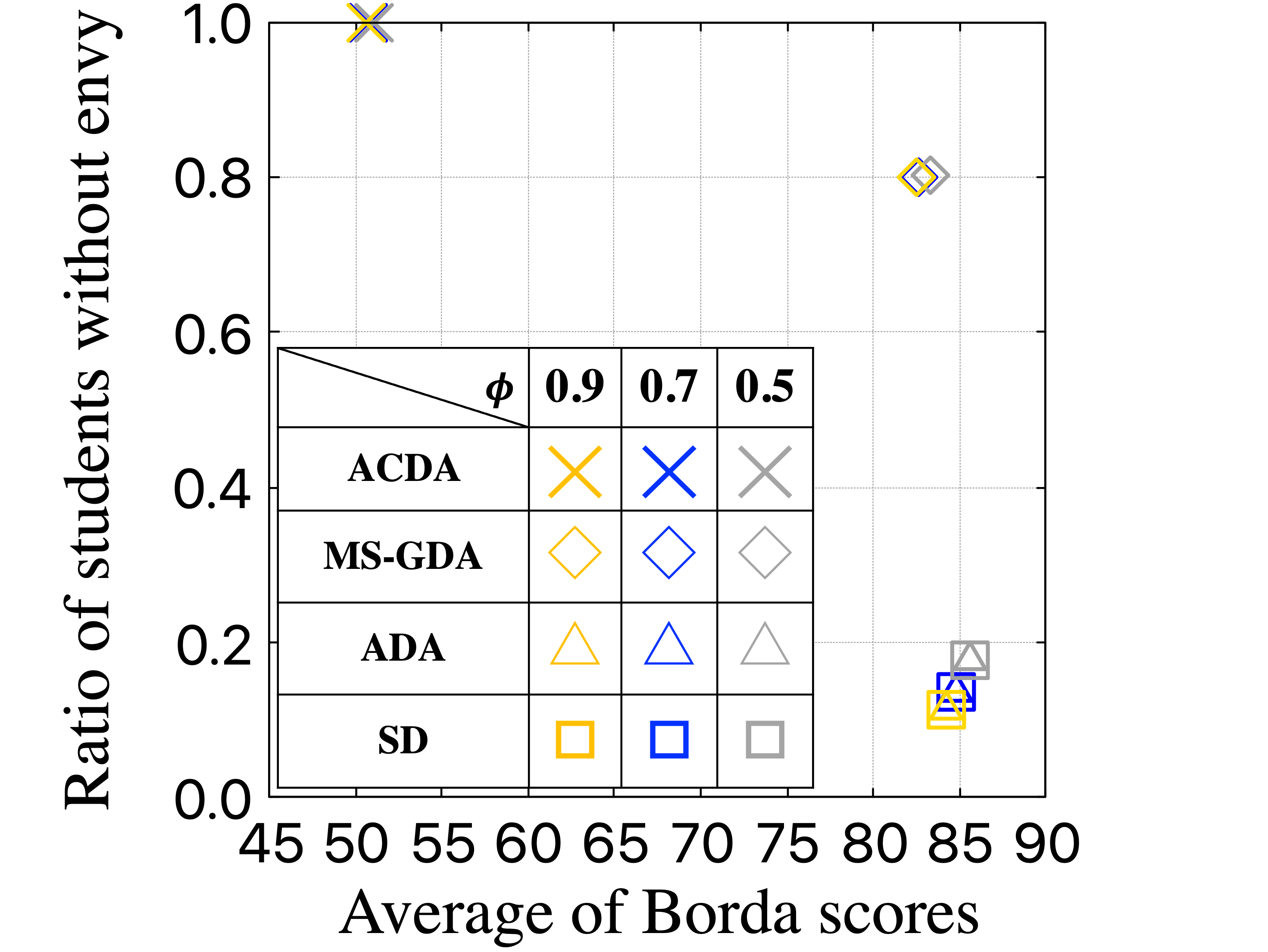}}
    \subfloat[Effect of $\phi$ ($ q = 800$)]
	{\includegraphics[scale = 0.23]{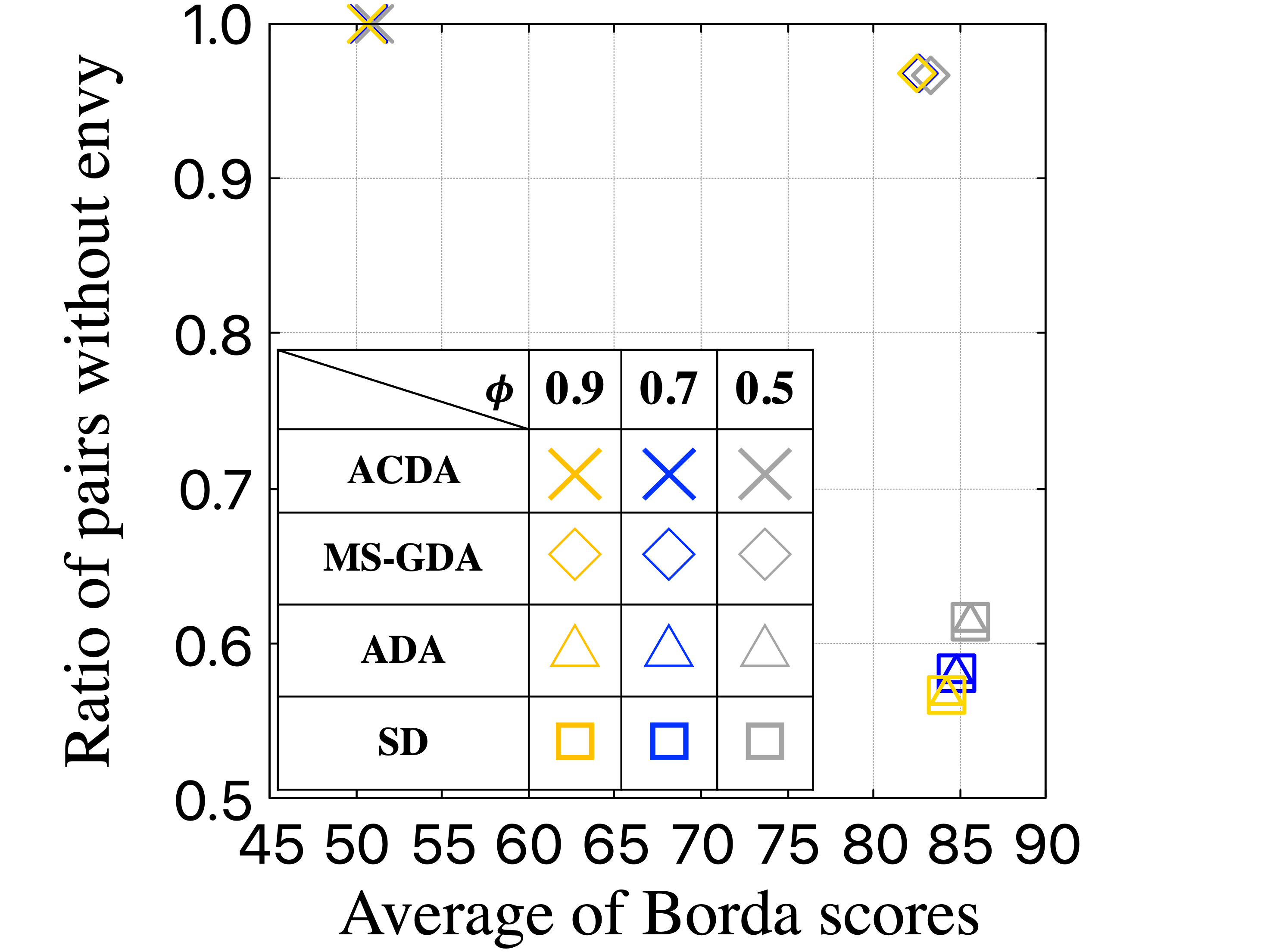}}
	\subfloat[Effect of $q$ ($\phi = 0.7$)]
	{\includegraphics[scale = 0.23]{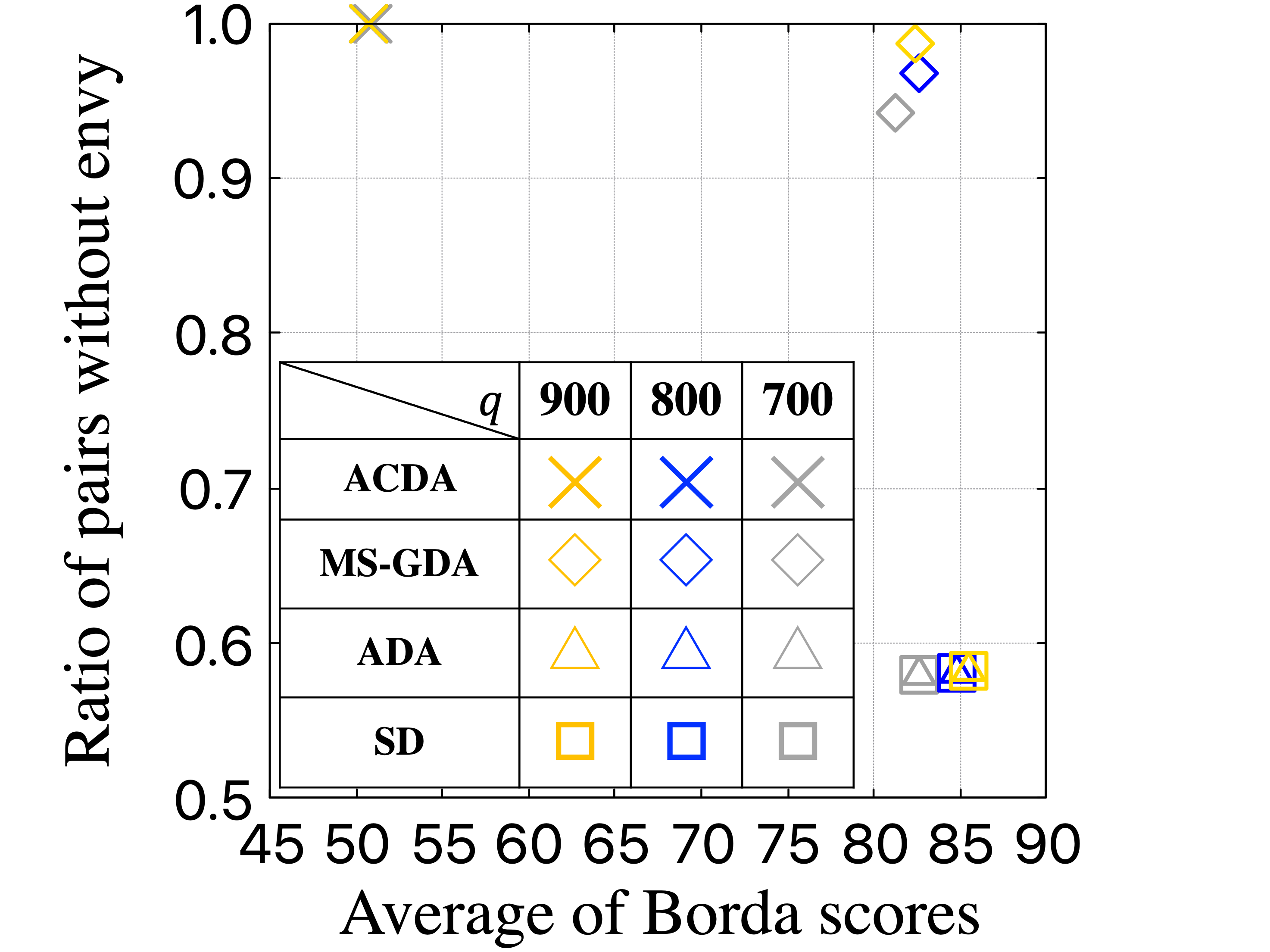}}
 \vspace{\trimbeforecaption}   
	\captionsetup{justification=centering}
	\caption{Trade-off between efficiency and fairness 
 (Market 1)}
	\label{fig:graph1}
\end{figure}
\begin{figure}[t]
	\captionsetup[subfigure]{aboveskip=5pt,belowskip=2pt}
	\centering
	\captionsetup{justification=centering}
    \subfloat[Effect of $\phi$ ($q = 450$)]
	{\includegraphics[scale = 0.23]{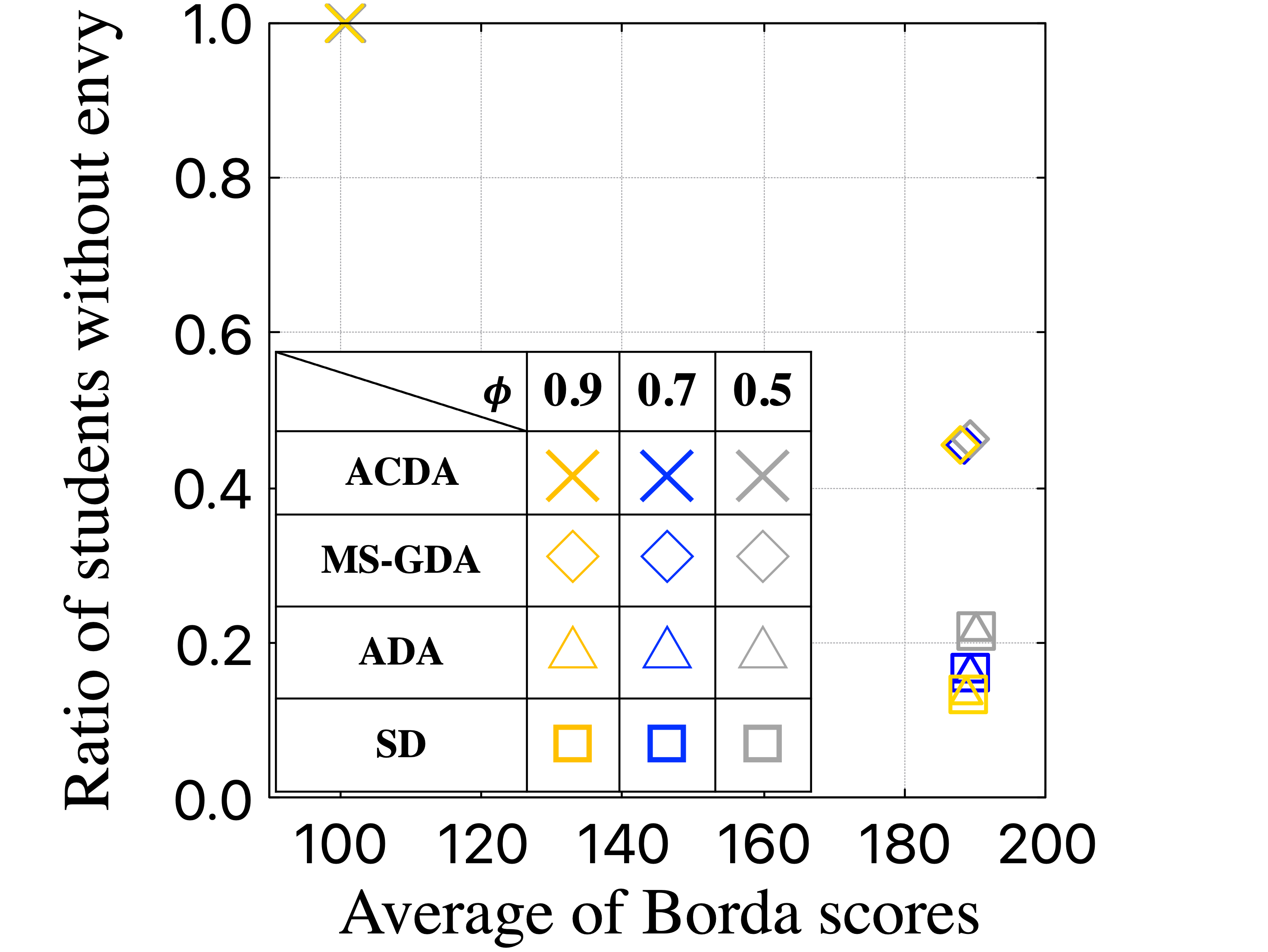}}
	\subfloat[Effect of $\phi$ ($q = 450$)]
	{\includegraphics[scale = 0.23]{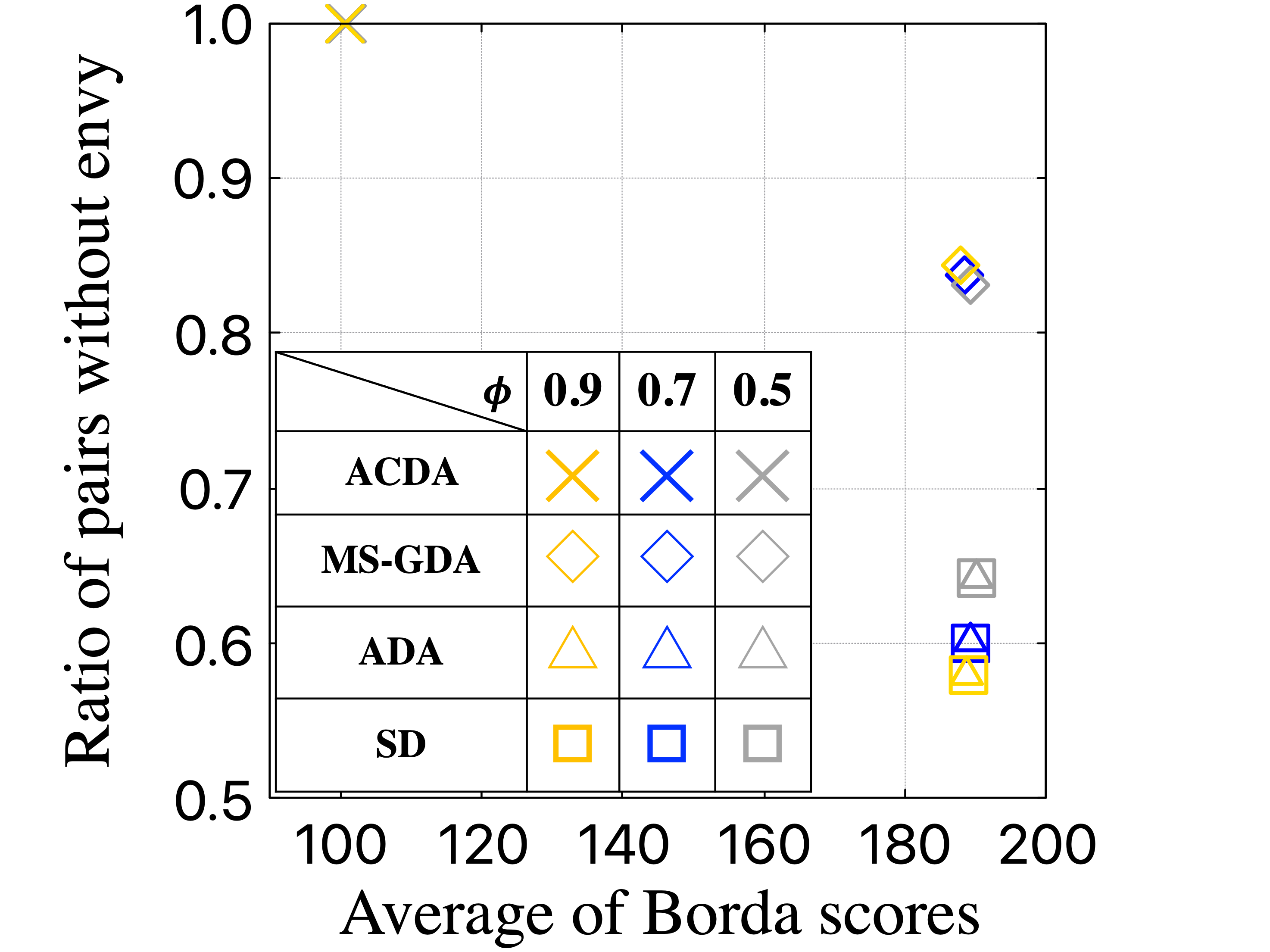}}
	\subfloat[Effect of $q$ ($\phi = 0.7$)]
	{\includegraphics[scale = 0.23]{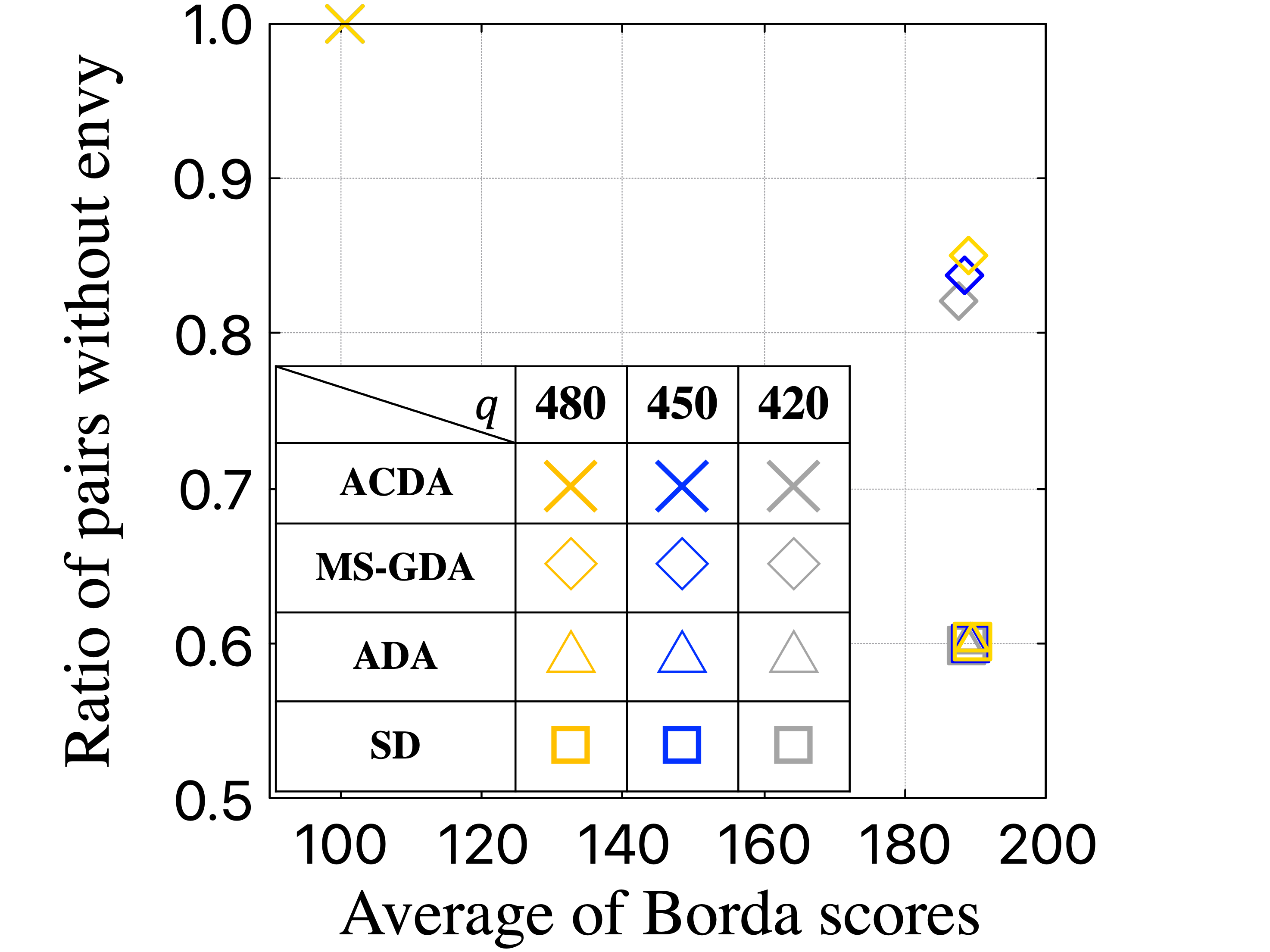}}
 \vspace{\trimbeforecaption}  
	\captionsetup{justification=centering}
	\caption{Trade-off between efficiency and fairness 
 (Market 2)}
	\label{fig:graph2}
\end{figure}

In both markets, 
student preferences are generated with the Mallows model 
\cite{drummond:ijcai:2013,lu:jmlr:2014,mallows1957non};
student preference $\succ_s$ is drawn with probability: 
$\frac{\exp (-\phi \cdot \delta(\succ_s,
	\succ_{\widehat{s}}))}{\sum_{\succ'_s} \exp(-\phi \cdot \delta(\succ'_s,
	\succ_{\widehat{s}}))}$.
Here, $\phi \in \mathbf{R}_{+}$ denotes a spread parameter, $\succ_{\widehat{s}}$ is a
central preference (uniformly randomly chosen from all possible
preferences).
Intuitively, student preferences are distributed around a central preference with spread parameter $\phi$. 
When $\phi= 0$, the Mallows model becomes identical to the uniform distribution 
and as $\phi$ increases, it quickly converges to the constant distribution that returns $\succ_{\widehat s}$. 
The preference of each college is uniformly randomly chosen.

To illustrate the trade-off between efficiency and fairness, we plot {the results} of the obtained matchings 
in a two-dimensional space, where the $x$-axis shows the average Borda scores of the students; if a student is assigned to her $i$-th choice college, her score is $m-i+1$. 
In Figures~\ref{fig:graph1}~(a) and \ref{fig:graph2}~(a),
the $y$-axis shows the ratio of students without justified envy, 
which counts the number of students who has no justified envy.
Except for these figures, 
the $y$-axis shows the ratio of the student pairs without justified envy, which counts the number of pairs of students without justified envy against each other. In both cases, the points located northeast are preferable. Each point represents the result of an average of 100 instances for one mechanism.

In Figures~\ref{fig:graph1}~(a)-(b) and 
\ref{fig:graph2}~(a)-(b), 
we vary $\phi$ in the Mallows model from $0.7$ to $0.9$. When $\phi$ becomes larger, the competition among students
becomes more fierce. 
We can see that 
MS-GDA strikes a good balance between efficiency and fairness.
Its Borda scores are very close to ADA/SD, and the ratio of students
without justified envy is about 80\% (Figure~\ref{fig:graph1}~(a))
and 46\% (Figure~\ref{fig:graph2}~(a)) for MS-GDA, 
while it is about 11\% to 17\% (Figure~\ref{fig:graph1}~(a))
and 11\% to 21\% (Figure~\ref{fig:graph2}~(a))
for ADA/SD.
Figures~\ref{fig:graph1}~(b) and \ref{fig:graph2}~(b) show the result of the same setting, but 
the $y$-axis shows the ratio of student pairs without justified envy. 
More than 95\% pairs (Figure~\ref{fig:graph1}~(b)) 
and 85\% pairs (Figure~\ref{fig:graph2}~(b)) 
have no justified envy in MS-GDA, while about 56\% to 61\% 
pairs (Figure~\ref{fig:graph1}~(b)) 
and 57\% to 64\% pairs (Figure~\ref{fig:graph2}~(b))
have no justified envy in ADA/SD. 
In Figure~\ref{fig:graph1}~(c),  we vary $q$ from 
$700$ to $900$. 
As the quota for non-rural colleges becomes 
less binding (thus, the distributional constraints become closer to an \Mnatural-convex set), MS-GDA becomes fairer. 
In Figure~\ref{fig:graph2}~(c),  we vary $q$ (the base-line 
regional quota) to 420, 450, and 480, as well as the additional flexible amount to 160, 100, and 40, respectively. 
As the additional flexible amount becomes smaller, 
the distributional constraints become closer to 
\Mnatural-convex set. Thus, MS-GDA becomes slightly fairer. 

\section{Conclusions and Future Works}
\label{sec:conclusions}
This paper developed a new mechanism called MS-GDA, which can work for any hereditary constraints and strikes a good balance
between fairness and efficiency. It uses GDA as a subroutine. 
We proved that MS-GDA is strategyproof, weakly nonwasteful, ML-fair, 
and on the Pareto frontier. 
We experimentally showed that MS-GDA is fairer than ADA, while 
it does not sacrifice the efficiency/students' welfare too much when 
distributional constraints are close to an \Mnatural-convex set. 
Our future works include 
evaluating the performance of 
MS-GDA in various real-life application domains. 

\section*{Acknowledgments}

 This work was partially supported by the JST ERATO  Grant Number JPMJER2301 and JSPS
 KAKENHI Grant Numbers JP20H00609, JP21H04979, Japan.
 
\bibliographystyle{abbrvnat}
\bibliography{matching}
\end{document}